\documentclass[onefignum,onetabnum]{siamart171218}

\usepackage[T1]{fontenc}            
\usepackage[utf8]{inputenc}         
\usepackage{microtype}              
\usepackage{amsfonts}
\DeclareUnicodeCharacter{00A0}{ }   
\usepackage[textwidth=3cm,textsize=footnotesize]{todonotes}

\usepackage{amsmath}
\usepackage{amssymb}
\usepackage{commath}          

\usepackage{mathtools}
\usepackage{thmtools}

\Crefname{subsection}{Section}{Sections}

\usepackage{verbatim}         

\usepackage{graphicx}
\usepackage{epstopdf}
\ifpdf
  \DeclareGraphicsExtensions{.eps,.pdf,.png,.jpg}
\else
  \DeclareGraphicsExtensions{.eps}
\fi

\renewcommand{\epsilon}{\varepsilon}    

\newcommand{\Xkf}{\widehat{X}}            
\newcommand{\Lin}{\mathcal{L}} 
\newcommand{\Lric}{\mathcal{R}}        
\newcommand{\Qt}{Q_\text{tu}}         

\newcommand{\R}{\mathbb{R}}             
\DeclareMathOperator{\neper}{e}         
\newcommand{\textsm}[1]{\text{\tiny{#1}}} 

\newcommand{\expec}{\mathbb{E}}           
\newcommand{\var}{\operatorname{Var}}     
\newcommand{\gauss}{\mathcal{N}}          

\newcommand{\T}{\mathsf{T}}                               
\newcommand{\trace}{\operatorname{tr}}                    
\DeclarePairedDelimiterX\inprod[2]{\langle}{\rangle}%
{ #1 , #2 }                                               
\newcommand{\mineig}{\lambda_\text{min}}                  
\newcommand{\maxeig}{\lambda_\text{max}}                  
\newcommand{\Id}{I}                                       

\declaretheorem[numberlike=theorem,Refname={Assumption,Assumptions}]{assumption}

\headers{On Stability of a Class of Non-Linear Filters}{T. Karvonen, S. Bonnabel, E. Moulines, and S. S\"{a}rkk\"{a}}

\title{On Stability of a Class of Filters for Non-Linear Stochastic Systems\thanks{Submitted to arXiv on June 10, 2020.
    \funding{Toni Karvonen was supported by the Aalto ELEC Doctoral School, the Foundation for Aalto University Science and Technology, and the Lloyd’s Register Foundation programme on data-centric engineering at the Alan Turing Institute, United Kingdom; Simo S\"{a}rkk\"{a} by the Academy of Finland; and Eric Moulines by the chair ``BayeScale - P.\ Laffitte''.}}}

\author{Toni~Karvonen\thanks{Department of Electrical Engineering and Automation, Aalto University, Finland and the Alan Turing Institute, United Kingdom
    (\email{tkarvonen@turing.ac.uk}).}
\and Silv\`{e}re~Bonnabel\thanks{Mines ParisTech, PSL Research University, Centre for Robotics, Paris, France
  (\email{silvere.bonnabel@mines-paristech.fr}).}
\and Eric~Moulines\thanks{Centre de Math\'{e}matiques Appliqu\'{e}es, \'{E}cole Polytechnique, Palaiseau, France
  (\email{eric.moulines@polytechnique.edu}).}
\and Simo~S\"{a}rkk\"{a}\thanks{Department of Electrical Engineering and Automation, Aalto University, Finland
  (\email{simo.sarkka@aalto.fi}).}
}

\usepackage{amsopn}

\ifpdf
\hypersetup{
  pdftitle={On Stability of a Class of Filters for Non-Linear Stochastic Systems},
  pdfauthor={T. Karvonen, S. Bonnabel, E. Moulines, and S. S\"{a}rkk\"{a}}
}
\fi

\begin{document}

\maketitle

\begin{abstract}
  This article develops a comprehensive framework for stability analysis of a broad class of commonly used continuous and discrete time-filters for stochastic dynamic systems with non-linear state dynamics and linear measurements under certain strong assumptions. The class of filters encompasses the extended and unscented Kalman filters and most other Gaussian assumed density filters and their numerical integration approximations. The stability results are in the form of time-uniform mean square bounds and exponential concentration inequalities for the filtering error. In contrast to existing results, it is not always necessary for the model to be exponentially stable or fully observed.
    We review three classes of models that can be rigorously shown to satisfy the stringent assumptions of the stability theorems. Numerical experiments using synthetic data validate the derived error bounds.
\end{abstract}



\section{Introduction}

Non-linear Kalman filters, such as the extended Kalman filter~(EKF), the derivative-free unscented Kalman filter (UKF), and other Gaussian integration filters, are fundamental tools widely used in estimating a latent time-evolving state from partial and noisy measurements in, for instance, automatic control, robotics, and signal processing~\cite{Sarkka2013}. These filters are local extensions to the classical linear Kalman filter for systems with non-linear state evolution or measurement equation.
Stability properties of the optimal linear Kalman filter are well understood, having been extensively studied since the 1960s in continuous~\cite{Bucy1967,Delyon2001,BishopDelMoral2017} and discrete time~\cite{DeystPrice1968,Jazwinski1970,AndersonMoore1981} settings. However, most systems of interest are non-linear, and non-linear extensions of the Kalman filter inherit no global optimality properties. Even though these filters tend to provide useful estimates, analyzing their stability is far from trivial.

This article analyzes stochastic stability, defined as time-uniform boundedness of the mean square filtering error, of a large class of extensions of the Kalman filter for systems with non-linear state dynamics and linear measurements.
Our main stability results, \Cref{thm:ConInEq,thm:ConInEqDisc}, provide time-uniform mean square filtering error bounds  and related exponential concentration inequalities for a large class of filters. The theorems significantly extend the recent results by Del Moral et al.\@~\cite{DelMoralKurtzmannTugaut2016} on the EKF for exponentially stable (i.e., contractive) and fully observed models. The most important extensions are of three types: (a) we formulate an apparently novel framework that allows for considering a large class of commonly used filters simultaneously, not just the EKF; (b) we do not require that the model be exponentially stable or fully observed; and (c) we also cover the discrete-time case. In practice, generalization (b) is enabled by introduction of a certain assumption on stability of the filter process. That this assumption is satisfied by some classes of models not exponentially stable nor fully observed, and thus beyond the scope of applicability of~\cite{DelMoralKurtzmannTugaut2016} even when the EKF is used, is demonstrated in \Cref{sec:examples}.
Even though the stability assumptions used in this article are extremely restrictive and essentially amount to stability of the filtering error process, we stress that they are, as far as we are aware of, the least restrictive among assumptions used in the literature that permit rigorous a priori assessment of stability.
A more detailed presentation of our contributions is provided in \Cref{sec:contributions}.

There are two underlying objectives in this article that are not present in most previous works:
\begin{itemize}
\item[(i)] Our stability analysis is \emph{general} and \emph{unified} in that the class of filters it applies to encompasses most non-linear Kalman filters commonly used in relatively low-dimensional applications, such as tracking.
Ensemble Kalman filters, useful in high-dimensional applications, have been recently analyzed in~\cite{deWiljesReichStannat2017,DelMoralKurtzmannTugaut2017}.
\item[(ii)] We require that stability be rigorously \emph{a priori verifiable} and the error bounds \emph{a priori computable}. This means that it should be possible to conclude that a filter is stable and compute mean square error bounds before the filter is run. 
Accordingly, three classes of models for which this is possible are reviewed in \Cref{sec:examples}.
These requirements are in contrast to much of the existing literature where the results rely on opaque and difficult-to-verify assumptions~\cite{XiongZhangChan2006,LiXia2012} or no example models are provided that can be rigorously shown to satisfy the assumptions~\cite{ReifGuntherYazUnbehauen1999,ReifGuntherYazUnbehauen2000,KlugeReifBrokate2010}.
\end{itemize}
The second objective is crucial if stability results are to be applied in practice and is in some contrast to earlier work where it is occasionally suggested that (a) having values of certain parameters, as computed when the filter is run, satisfy the bounds or conditions required for stability allows for concluding that a stochastic filter is stable (e.g.,~\cite[p.\ 716]{ReifGuntherYazUnbehauen1999} and~\cite[p.\ 244]{XiongLiuZhang2009}), which is problematic if one is considering stability in mean square sense since the conditions are validated only for one particular trajectory, though more acceptable in the deterministic setting~\cite[p.\ 566]{bonnabel2015contraction}, or that (b) the true state can be assumed to remain in a compact set (see \cite[Theorem 4.1]{ReifGuntherYazUnbehauen1999} and \cite{BishopJensfelt2008}).
A consequence of this is that we work only with linear measurement models.
However, it should be noted that, out of necessity, many models that have been previously used in demonstrating stability results have linear measurements; see for example the model examined in~\cite{ReifGuntherYazUnbehauen1999,XiongZhangChan2006}.

\subsection{Previous Work and Technical Aspects}

A Kalman filter\footnote{Later on, when we want to refer specifically to continuous-time Kalman filters, we use the term \emph{Kalman--Bucy filter}.} or its non-linear extension provides, at time $t \geq 0$, an online estimate $\Xkf_t$ constructed out of a potentially partial and noisy measurement sequence $\{Y_s\}_{s=0}^t$ of the true latent state~$X_t$ of a dynamic system. The estimates are typically accompanied with positive-semidefinite matrices $P_t$, which are estimates of covariances of the estimation errors $E_t = X_t - \Xkf_t$. These matrices and the associated gain matrices $K_t$ are computed from a Riccati-type differential equation. Stability of extensions of the Kalman filter for non-linear systems can be analyzed either in a deterministic or stochastic setting. In the former case, the state dynamics and measurements are noiseless and the positive-semidefinite matrices $Q$ and $R$, which in the stochastic case would be covariances of Gaussian state and measurement noise terms, are tuning parameters. The goal is to prove that the estimation error converges to zero as $t \to \infty$. In the stochastic setting it cannot be expected that the error vanishes, and one instead (for example) attempts to prove time-uniform upper bounds or concentration inequalities for the mean square estimation error, $\expec(\norm[0]{E_t}^2)$.

There is a large body of literature on stability properties, both in continuous and discrete time, of the EKF as a non-linear observer~\cite{ReifSonnemannUnbehauen1996, BabacanOzbekEfe2008,  BoutayebRafaralahyDarouach1997, ReifSonnemannUnbehauen1998, ReifUnbehauen1999, Krener2003, MareeImslandJouffroy2016, bonnabel2015contraction, barrau2017invariant} and in the stochastic setting~\cite{ReifGuntherYazUnbehauen1999, ReifGuntherYazUnbehauen2000, Rapp2004, BishopJensfelt2008, KlugeReifBrokate2010,DelMoralKurtzmannTugaut2016} and of the UKF and related filters~\cite{XiongZhangChan2006, XiongZhangChan2007, WuHuHu2007, XiongChanZhang2007, XiongLiuZhang2009, LiXia2012,KarvonenSarkka2016}. 
However, majority of these articles attempt to be \emph{too} general, which often results in the use of assumptions that are effectively impossible to verify, especially before the filter is actually run, or in a lack of discussion on and examples of models for which the assumptions hold.
There are two principal sources of difficulty in the stability analysis of non-linear filters:

\paragraph{\textbf{Time-uniform bounds on $\pmb{P_t}$}} Analysis in most of the above articles is similar to the standard stability analysis for linear models~\cite{DeystPrice1968,Jazwinski1970} in that use is made of the Lyapunov function \sloppy{${V_t = E_t^\T P_t^{-1}E_t}$} or its variants. Once stability results have been obtained for $V_t$, time-uniform bounds on $P_t$ are necessary for concluding stability of the filter. While in the linear case $P_t$ is deterministic and bounds on this matrix follow from results on Riccati equations under certain observability and controllability conditions, in the non-linear case the local structure of most Kalman filters, arising from linearizations of some sort around the estimated trajectory, introduces a dependency of $P_t$ on the measurements and estimates. 
Consequently, the behavior of $P_t$ is difficult, if not impossible, to anticipate and control for most non-linear models and filters.

\paragraph{\textbf{Model non-linearity}} If the system is non-linear, stability analysis of a Kalman filter necessarily involves analyzing non-linear (stochastic) differential equations. This is obviously much more involved than analysis of linear differential equations. As such, the approach taken in many articles is to assume that the error associated to the linearization method used in a particular non-linear Kalman filter is ``small''. This allows for deriving a linear differential inequality for the Lyapunov function that can be easily controlled. 

When not outright assumed, boundedness of $P_t$ has been addressed essentially in two ways. If the system is fully observed, that is, $\dif Y_t = X_t \dif t + R^{1/2} \dif V_t$, there is hope for the Riccati equation to be well behaved despite the fact it depends on $\Xkf_t$ since, essentially, the quadratic correction term in the Riccati equation prevents $\Xkf_t$ and $P_t$ from drifting indefinitely; see~\cite[Section IV]{KlugeReifBrokate2010} and \cite[Section 4]{LiXia2012} for the discrete and~\cite{Karvonen2018} for the continuous-time case. Alternatively, one can consider certain difficult-to-verify non-linear extensions of the standard observability and controllability conditions~\cite{BarasBensoussanJames1988, ReifGuntherYazUnbehauen1999,ReifGuntherYazUnbehauen2000,Krener2003}. Another situation of interest is when the estimates are explicitly known to remain in a bounded region of the state space, which provides some control over the estimate-dependent terms in the Riccati equation and limits the possible values of $P_t$. See for example~\cite{BishopJensfelt2008} where stochastic stability of the EKF in a robotics application is considered. Model non-linearity is often dealt with by enforcing Lipschitz-type bounds on the remainder related to the particular linearization method used~\cite{ReifGuntherYazUnbehauen1999, ReifGuntherYazUnbehauen2000} or by assuming boundedness of certain residual-correcting random matrices~\cite{XiongZhangChan2006,XiongLiuZhang2009,LiXia2012}. However, such assumptions tend to be difficult to verify.

\subsection{Contributions}\label{sec:contributions}

This article follows the approach taken recently by Del Moral et al.\ \cite{DelMoralKurtzmannTugaut2016}. They study stochastic stability of the extended Kalman--Bucy filter by directly considering the squared error $\norm[0]{E_t}^2$ for which they derive stochastic differential inequalities that in turn establish time-uniform mean square error bounds and exponential concentration inequalities. However, the class of systems they consider is very restricted as they need to assume that the state is fully observed and the dynamic model, as specified by the drift function, is exponentially stable (i.e., the deterministic homogeneous differential equation $\partial_t x_t = f(x_t)$ defined by the drift $f$ is exponentially stable).
In \Cref{thm:ConInEq,thm:ConInEqDisc} we introduce several significant generalizations and improvements to the results in~\cite{DelMoralKurtzmannTugaut2016}:
\begin{enumerate}
\item We consider a broad class, defined in \Cref{sec:filters}, of generic Kalman-type filters for continuous-time non-linear systems when the measurements are linear; see~\eqref{eq:model} for the model. As demonstrated in \Cref{sec:kalmanBucy}, this class of filters contains many commonly used filters, including the extended Kalman--Bucy filter and the more recent Gaussian integration filters such as the unscented Kalman--Bucy filter~\cite{julier1997new,Sarkka2007} and the Gauss--Hermite filter~\cite{WuHuWuHu2006,SarkkaSarmavuori2013}. This unified framework is exceedingly convenient as every filter does not have to be analyzed individually. There have been prior attempts at establishing a unified stability analysis~\cite{WuHuHu2007, XiongChanZhang2007, KarvonenSarkka2016}, but the formulations are somewhat unnatural, being in terms of certain residual terms that are difficult to control.

\item Unlike in~\cite{DelMoralKurtzmannTugaut2016}, the system is not explicitly required to be exponentially stable or fully observed. 
While still very stringent, the assumption we use is satisfied by a larger class of models (for discussion on the assumptions, see \Cref{sec:stability}). Two model and filter classes which do not require exponential stability or full observability are reviewed in \Cref{sec:covInf,sec:IntVel}.

\item Although our main focus is on continuous-time systems, \Cref{sec:discrete} contains analogous results for the discrete-time case. The discrete case is instructive in demonstrating rigorously that under appropriate conditions a non-linear Kalman filter improves upon the trivial estimator $\Xkf_{Y,k} = Y_k$. This is discussed in \Cref{sec:meas-acc}.

\item \Cref{sec:numerics} contains two numerical examples that demonstrate conservativeness of the derived mean square error bounds. Unlike in much of the literature (e.g.,~\cite[Section~V]{ReifGuntherYazUnbehauen1999} and~\cite[Section~5]{ReifGuntherYazUnbehauen2000}), we can verify beforehand that the example models satisfy the stability assumptions.
  
\end{enumerate}
Although some elements of the proofs are similar to those in~\cite{DelMoralKurtzmannTugaut2016}, inclusion of complete and self-contained proofs is necessary because our adoption of a general class of filters introduces modifications, some of the constants involved are different, and also the discrete case, for which the analysis has not been carried out before, is considered.


\section{Non-Linear Systems and Filtering} \label{sec:filters}

This section introduces the continuous-time stochastic dynamic systems and the class of stochastic Kalman--Bucy filters the results in \Cref{sec:stability} apply to.
A number of prominent members of this filter class are also given.
Discrete-time systems and filters are discussed in \Cref{sec:discrete}.

\subsection{Logarithmic Norms and Lipschitz Constants}

The smallest and largest eigenvalues of a symmetric real matrix $A$ are $\mineig(A)$ and $\maxeig(A)$. The \emph{logarithmic norm} $\mu(A)$ of a square matrix $A \in \R^{d \times d}$ is \sloppy{${\mu(A) = \frac{1}{2} \maxeig (A + A^\T)}$},
coinciding with $\maxeig(A)$ when $A$ is symmetric. We also define the quantity \sloppy{${\nu(A) = \frac{1}{2} \mineig (A + A^\T) = -\mu(-A)}$}.
Basic results that we repeatedly use are
\begin{equation*}
\nu(A) \norm[0]{x}^2 \leq \inprod{x}{Ax} = x^\T Ax \leq \mu(A) \norm[0]{x}^2
\end{equation*}
for any $x \in \R^d$ and the ``triangle inequalities''
\begin{align*}
\maxeig(A + B) &\leq \maxeig(A) + \maxeig(B), & \mineig(A+B) &\geq \mineig(A) + \mineig(B), \\
\mu(A + B) &\leq \mu(A) + \mu(B), & \nu(A+B) &\geq \nu(A) + \nu(B).
\end{align*}
For a positive-semidefinite $B$, recall also the trace inequality~\cite[Chapter 8]{Bernstein2009}
\begin{equation}\label{eq:traceBounds}
\nu(A) \trace(B) \leq \trace(AB) \leq \mu(A) \trace(B)
\end{equation}
for any square matrix $A$ and its special case $\mineig(A) \trace(B) \leq \trace(AB) \leq \maxeig(A) \trace(B)$ for a symmetric $A$. See~\cite{Strom1975,Soderlind2006} for detailed reviews of the logarithmic norm.

Let $g \colon \R^d \to \R^d$ be differentiable and \sloppy{${[J_g]_{ij} = \partial g_i / \partial z_j}$} its Jacobian matrix. The \emph{Lipschitz constant} of $g$ is $\norm[0]{J_g} = \sup_{x \in \R^d} \norm[0]{J_g(x)}$, where the matrix norm is the norm induced by the Euclidean norm (i.e., the spectral norm). This constant satisfies $\norm[0]{g(x) - g(x')} \leq \norm[0]{J_g} \norm[0]{x - x'}$ for any $x, x' \in \R^d$. If $\norm[0]{J_g} < \infty$, the function $g$ is \emph{Lipschitz}. The \emph{logarithmic Lipschitz constants} of $g$ are
\begin{equation*}
N(g) = \inf_{z \in \R^d} \nu[J_g(z)] \hspace{0.3cm} \text{and} \hspace{0.3cm} M(g) = \sup_{z \in \R^d} \mu[J_g(z)].
\end{equation*}
These constants satisfy
\begin{equation}\label{eq:lipBounds}
N(g) \norm[0]{x-x'}^2  \leq \inprod[\big]{x-x'}{g(x) - g(x')} \leq M(g) \norm[0]{x-x'}^2 ,
\end{equation}
for any  $x,x' \in \R^d$. Note that $M(g) \leq \norm[0]{J_g}$~\cite[Proposition~3.1]{Soderlind2006}.

\subsection{System Description}

We consider systems of stochastic differential equations of the form
\begin{subequations}\label{eq:model}
\begin{align}
\dif X_t &= f(X_t) \dif t + Q^{1/2} \dif W_t, \label{eq:modelState} \\
\dif Y_t &= H X_t \dif t + R^{1/2} \dif V_t,
\end{align}
\end{subequations}
where $X_t \in \R^{d_x}$ is the latent \emph{state} evolving according to a continuously differentiable and potentially non-linear \emph{drift} $f \colon \R^{d_x} \to \R^{d_x}$. We assume that the drift is Lipschitz (i.e., $\norm[0]{J_f} < \infty$) and that its Jacobian is bounded in logarithmic norm:
\begin{equation}\label{eq:JacobianBound}
-\infty < N(f) = \inf_{x \in \R^{d_x}} \nu[J_f(x)] \quad \text{ and } \quad M(f) = \sup_{x \in \R^{d_x}} \mu[J_f(x)] < \infty.
\end{equation}
These conditions ensure that the state and the filters defined later in this section remain almost surely bounded in finite time. The \emph{measurements} $Y_t \in \R^{d_y}$ are obtained linearly through a measurement model matrix $H \in \R^{d_y \times d_x}$. Both the state and measurements are disturbed by independent multivariate Brownian motions $W_t \in \R^{d_x}$ and $V_t \in \R^{d_y}$ multiplied by positive-definite noise covariance matrices $Q \in \R^{d_x \times d_x}$ and $R \in \R^{d_y \times d_y}$. The state is initialized from $X_0 \sim \gauss(\mu_0, \Sigma_0)$ for some mean $\mu_0 \in \R^{d_x}$ and a positive-definite covariance $\Sigma_0 \in \R^{d_x \times d_x}$.

The results of this article remain valid if the time-invariant function~$f$ and matrices $H$, $Q$, and $R$ in \eqref{eq:model} are replaced with time-varying versions that satisfy appropriate regularity and uniform boundedness conditions.
For instance, with a time-varying drift $f_t$ the assumptions~\eqref{eq:JacobianBound} become
\begin{equation*}
  -\infty < \inf_{t \geq 0} \, \inf_{x \in \R^{d_x}} \nu[J_{f_t}(x)] \quad \text{ and } \quad \sup_{t \geq 0} \, \sup_{x \in \R^{d_x}} \mu[J_{f_t}(x)] < \infty.
\end{equation*}
Later in Theorem~\ref{thm:ConInEq} the crucial assumption~\eqref{eq:contractivityCondition} would be replaced with
\begin{equation*}
  \sup_{ t \geq T} M(f_t - P_t S_t) \leq - \lambda < 0,
\end{equation*}
where $S_t = H_t R_t^{-1} H_t$ is a time-varying version of the matrix $S$ in~\eqref{eq:definition-K-S}, and $\trace(Q)$ and $\trace(S)$ on the right-hand side of the mean square error bound~\eqref{eq:MainMeanSquare} would become $\sup_{t \geq T} \trace(Q_t)$ and $\sup_{t \geq T} \trace(S_t)$.
We work in the time-invariant setting in order to keep the notation simpler.

\subsection{The Extended Kalman--Bucy Filter} \label{sec:ekf}

The extended Kalman--Bucy filter (EKF) is a classical method for computing estimates $\Xkf_t$ of the latent state $X_t$ of the system~\eqref{eq:model}. The EKF is based on local first-order linearizations around the estimated states. It is defined by the equations
\begin{subequations}\label{eq:ekf}
\begin{align}
\dif \Xkf_t &= f(\Xkf_t) \dif t + P_t H^\T R^{-1} \big( \dif Y_t - H \Xkf_t \dif t \big), \\
\partial_t P_t &= J_f(\Xkf_t) P_t + P_t J_f(\Xkf_t)^\T + \Qt - P_t S P_t, \label{eq:riccati-ekf}
\end{align}
\end{subequations}
where
\begin{equation}
\label{eq:definition-K-S}
K_t = P_t H^\T R^{-1} \hspace{0.5cm} \text{and} \hspace{0.5cm} S = H^\T R^{-1} H,
\end{equation}
the former of which are known as \emph{Kalman gain} matrices.
Equation \eqref{eq:riccati-ekf} governing evolution of $P_t$ is known as the (non-linear) \emph{Riccati equation}. The matrix $\Qt$ is a positive-definite matrix that does not have to be equal to $Q$, the state noise covariance, in which case we can speak of \emph{tuning} this matrix~\cite{BolognaniTubianaZigliotto2003}. The rest of this section introduces a framework for generalized Kalman-type filters similar in structure to the EKF and amenable to a unified stability analysis.

\subsection{A Class of Generic Filters for Non-Linear Systems}\label{sec:filters-generic}

A \emph{filter} computes a quantity $\Xkf_t \in \R^{d_x}$ that is used as an estimate of the latent state $X_t$. We consider \emph{generic} filters defined as
\begin{equation}\label{eq:filter}
\dif \Xkf_t = \Lin_{\Xkf_t,P_t} (f) \dif t + P_t H^\T R^{-1} \big( \dif Y_t - H \Xkf_t \dif t \big),
\end{equation}
where $\Lin_{x,P}$ is a parametrized linear functional, to be discussed in detail below, that maps functions $g \colon \R^{d_x} \to \R^{d_x}$ to $\R^{d_x}$ and the matrices $P_t \in \R^{d_x \times d_x}$ are user-specified, can depend on all the system parameters as well as all preceding measurements and state estimates, and are measurable with respect to the $\sigma$-algebra $\mathcal{F}_t = \sigma(Y_s, s \leq t)$ generated by the measurements. An assumption that $P_t$ is sufficiently regular and well-behaved and Lipschitzianity in $x$ and $P$ of $\Lin_{x,P}(f)$ guarantee the existence of a unique solution to~\eqref{eq:filter}. Explicit examples of filters follow in Section~\ref{sec:kalmanBucy}.
We initialize the filter~\eqref{eq:filter} with a deterministic $\Xkf_0 = \hat{x}_0 \in \R^{d_x}$ and a positive-definite $P_0 \in \R^{d_x \times d_x}$. These do not have to be equal to $\mu_0$ or $\Sigma_0$, respectively, the mean and covariance of the initial state $X_0$.
In \Cref{sec:stability} we will see that, as long as they remain uniformly bounded, the construction of the matrices $P_t$ does not substantially affect our analysis.

The linear functional $\Lin_{x,P}$ is parametrized by $x \in \R^{d_x}$ and $P \in \R^{d_x \times d_x}$ and it is required that the functional
\begin{enumerate}
\item[(i)] is Lipschitz (and hence continuous) in the parameters $x$ and $P$ in the sense that $\Lin_{x,P}(g)$ is a Lipschitz function from $\R^{d_x} \times \R^{d_x \times d_x}$ to $\R^{d_x}$ for any fixed Lipschitz function $g \colon \R^{d_x} \to \R^{d_x}$;
\item[(ii)] satisfies $\Lin_{x,P}(g) = g(x)$ for any $x$ and $P$ if $g(x) = Ax + b$ for some $A \in \R^{d_x \times d_x}$ and $b \in \R^{d_x}$.
\end{enumerate}
Note that it is \emph{not} necessary for $\Lin_{x,P}$ to depend on $P$; a prototypical example is the standard point evaluation functional $\Lin_{x,P}^\textsm{EKF}(g) = g(x)$ for any $P$. Another functional that is used in this article is the Gaussian integration functional
\begin{equation*}
\begin{split}
  \Lin_{x,P}^\textsm{ADF}(g) &= \int_{\R^{d_x}} g(z) \gauss(z \mid x, P)\dif z \\
  &\coloneqq (2\pi)^{-d_x/2} \det (P)^{-1/2} \int_{\R^{d_x}} g(z) \exp\bigg( \!\! - \frac{1}{2} (z - x)^\T P^{-1} (z - x) \bigg) \dif z.
\end{split}
\end{equation*}
The above requirements on $\Lin_{x,P}$ are usually easily verifiable and non-restrictive. The following less straightforward assumption is crucial to the stability analysis in \Cref{sec:stability}.

\begin{assumption}\label{ass:Lt} For any differentiable $g \colon \R^{d_x} \to \R^{d_x}$ with finite $N(g)$ and $M(g)$ there is a constant $C_g \geq 0$, which varies continuously with $M(g)$ and $N(g)$, such that
\begin{equation*}
\inprod[\big]{x-\tilde{x}}{g(x) - \Lin_{\tilde{x},P}(g)} \leq M(g) \norm[0]{x-\tilde{x}}^2 + C_g \trace(P)
\end{equation*}
for any  $x,\tilde{x} \in \R^{d_x}$ and  $P \in \R^{d_x \times d_x}$.
\end{assumption}
Since $\inprod{x-{\tilde{x}}}{g(x)-g({\tilde{x}})} \leq M(g) \norm[0]{x-{\tilde{x}}}^2$ by \eqref{eq:lipBounds}, what the above assumption essentially entails is that $\Lin_{\tilde{x},P}(g)$ cannot deviate too much from $g(\tilde{x})$ and that magnitude of their difference is controlled by the size of $P$.

The class of filters of the form~\eqref{eq:filter} that use a linear functional satisfying \Cref{ass:Lt} is very large. It encompasses, for example, the extended Kalman--Bucy filter and Gaussian assumed density filters and their most popular numerical integration approximations. Next we review a few such examples, demonstrating in the process that \Cref{ass:Lt} is indeed reasonable and fairly natural.

\subsection{Kalman--Bucy Filters for Continuous-Time Non-Linear Systems}\label{sec:kalmanBucy}

A Kalman--Bucy filter for the model~\eqref{eq:model} computes approximations $\Xkf_t$ and $P_t$, latter of which is called \emph{error covariance} in this setting, to the conditional filtering means and covariances $\expec( X_t \mid \mathcal{F}_t)$ and $\var( X_t \mid \mathcal{F}_t)$, respectively. 
It is usually difficult to derive tractable expressions for these quantities unless $f$ is affine.
A \emph{generalized} Kalman--Bucy filter for the model~\eqref{eq:model} is
\begin{subequations}\label{eq:kalmanbucy}
\begin{align}
\dif \Xkf_t &= \Lin_{\Xkf_t,P_t}(f) \dif t + P_t H^\T R^{-1} \big( \dif Y_t - H \Xkf_t \dif t \big), \label{eq:KFestimate} \\
\partial_t P_t &= \Lric_{\Xkf_t,P_t}(f) + \Lric_{\Xkf_t, P_t}(f)^\T + \Qt - P_t S P_t, \label{eq:riccati}
\end{align}
\end{subequations}
where the linear functional $\Lric_{x,P}$ maps functions to $d_x \times d_x$ matrices. A unique solution to~\eqref{eq:riccati} exists if $\Lric_{x,P}(f)$ is Lipschitz in $x$ and $P$. This holds typically when the Jacobian of $f$ satisfies $\norm[0]{J_f(x) - J_f(x')}\leq L \norm[0]{x-x'}$ for some $L < \infty$ and all $x,x' \in \R^{d_x}$. Examples of commonly used $\Lric_{x,P}$ appear below. As in the case of the EKF, we call~\eqref{eq:riccati} a Riccati equation and $\Qt$ is a positive-definite tuning matrix. As we shall see, proper tuning (in practice, inflation) is often necessary to induce provable stability of a Kalman--Bucy filter. Next we provide three examples of classical Kalman--Bucy filters of the form~\eqref{eq:kalmanbucy} that satisfy the assumptions in \Cref{sec:filters-generic}.

\subsubsection{Extended Kalman--Bucy Filter}

By selecting $\Lin_{x,P}(g) = \Lin_{x,P}^\textsm{EKF}(g) = g(x)$ and $\Lric_{x,P}(g) = \Lric_{x,P}^\textsm{EKF}(g) \coloneqq J_g(x) P$ we observe that the EKF in~\eqref{eq:ekf} is an example of a generalized Kalman--Bucy filter. Furthermore, Assumption~\ref{ass:Lt} is trivially satisfied by $\Lin_{x,P}^\textsm{EKF}$ with $C_g = 0$ for any function $g$. 

\subsubsection{Gaussian Assumed Density Filters}\label{sec:slf}

In \emph{Gaussian assumed density filters}~\cite{ItoXiong2000}, the point evaluations of the model functions and their Jacobians in the EKF are replaced with Gaussian expectations with mean $\Xkf_t$ and variance $P_t$. That is,
\begin{equation} \label{eq:L-slf}
\Lin_{x,P}(g) = \Lin_{x,P}^\textsm{ADF}(g) = \expec_{\gauss(x,P)}(g) \coloneqq \int_{\R^{d_x}} g(z) \, \gauss(z \mid x, P) \dif z
\end{equation}
and 
\begin{equation*}
  \Lric_{x,P}(g) = \Lric_{x,P}^\textsm{ADF}(g) \coloneqq \expec_{\gauss(x,P)}(J_g) P = \bigg( \int_{\R^{d_x}} J_g(z) \, \gauss(z \mid x, P) \dif z \bigg) P, 
\end{equation*}
where the integrals are element-wise.
Both of the basic properties  required of $\Lin_{x,P}$ in \Cref{sec:filters} hold. It can be shown that Assumption~\ref{ass:Lt} holds with $C_g = M(g) - N(g) \geq 0$; the straightforward proof is presented in Appendix~\ref{app:int-quad}.

\subsubsection{Gaussian Integration Filters}\label{sec:ukf}

Gaussian expectations required in implementation of the Gaussian assumed density filter  are typically unavailable in closed form, necessitating the use of numerical integration formulas. We call such filters \emph{Gaussian integration filters}. Popular alternatives include fully symmetric formulas, such as the ubiquitous unscented transform~\cite{julier1997new,Sarkka2007}, and tensor-product rules~\cite{WuHuWuHu2006,SarkkaSarmavuori2013}.

A Gaussian integration filter replaces the Gaussian expectations occurring in the Gaussian assumed density filter with numerical cubature approximations
\begin{equation}\label{eq:cubature}
\Lin_{x,P}^\textsm{int}(g) = \sum_{i=1}^n w_i g\big(x + \sqrt{P} \xi_i\big) \approx \expec_{\gauss(x,P)}(g),
\end{equation}
where $\xi_1,\ldots,\xi_n \in \R^{d_x}$ and $w_1,\ldots,w_n \in \R$ are user-specified \emph{unit sigma-points} and \emph{weights}, respectively, and $\sqrt{P}$ is some form of symmetric matrix square root of~$P$. 
The integral $\expec_{\gauss(x,P)}(J_g) P$ in the assumed density filter is replaced with $\Lric_{x,P}^\textsm{int}(g) = \sum_{i=1}^n w_i g\big(x + \sqrt{P} \xi_i\big)\xi_i^\T \sqrt{P}$, which makes use of Stein's identity
\begin{equation*}
\expec_{\gauss(x,P)}(J_g) P = \int_{\R^{d_x}} g(z) \big( z - x \big)^\T \gauss(z \mid x, P) \dif z.
\end{equation*}
Obviously, it is not necessary to use the same numerical integration scheme in $\Lin_{x,P}^\textsm{int}$ and $\Lric_{x,P}^\textsm{int}$.
In Appendix~\ref{app:int-quad} it is shown that \Cref{ass:Lt} holds with $C_g = M(g) - N(g)$ if the weights are non-negative and
\begin{equation}\label{eq:quadraticExact}
\Lin_{x,P}^\textsm{int}(p) = \expec_{\gauss(x,P)}(p)
\end{equation}
whenever $p \colon \R^d \to \R$ is a $d_x$-variate polynomial of total degree at most two. Among many other filters,~\eqref{eq:quadraticExact} is satisfied by the aforementioned Kalman--Bucy filters based on the unscented transform and Gaussian tensor-product rules. Filters that do not satisfy this assumption include kernel-based Gaussian process cubature filters~\cite{SarkkaEtal2016, PruherStraka2018}.

\subsubsection{On Ensemble Kalman--Bucy Filters}

The ensemble Kalman--Bucy filter for non-linear systems (e.g.,~\cite{TongMajdaKelly2016,deWiljesReichStannat2017}) is closely related to a Gaussian integration filter. The ensemble filter uses time-varying empirical estimate operators
\begin{equation*}
  \Lin_{x,P}^\textsm{int}(g, t) = \frac{1}{n} \sum_{i=1}^n g(\xi_{i,t}) \quad \text{ and } \quad \Lric_{x,P}^\textsm{int}(g, t) = \frac{1}{n-1} \sum_{i=1}^n (g(\xi_{i,t}) - x)(\xi_{i,t} - x)^\T,
\end{equation*}
where the time-varying and random sigma-points $\xi_{i,t}$ obey a differential equation derived from the model.
Without modifications our results do \emph{not} apply to filters of this type because $\Lin_{x,P}^\textsm{int}(\cdot,t)$ does not necessarily satisfy~\eqref{eq:quadraticExact} for second-degree polynomials and all $t \geq 0$.


\section{Stability of Kalman--Bucy Filters}\label{sec:stability}

The main result, \Cref{thm:ConInEq}, of this article contains an upper bound on the mean square filtering error and an associated exponential concentration inequality.
Similar exponential concentration inequalities have previously appeared in~\cite{DelMoralKurtzmannTugaut2016} for the extended Kalman--Bucy filter and in~\cite{DelMoralKurtzmannTugaut2017} for the ensemble Kalman--Bucy filter. See also~\cite{DelMoralTugaut2018,BishopDelMoral2017} for work regarding the linear case and~\cite{deWiljesReichStannat2017} for analysis, somewhat similar to ours, for ensemble Kalman--Bucy filters.

\Cref{thm:ConInEq} is based on the evolution equation
\begin{equation}\label{eq:filtering-error-evolution}
\dif E_t = \big[ f(X_t) - \Lin_{\Xkf_t,P_t}(f) - P_t S (X_t - \Xkf_t) \big] \dif t + Q^{1/2} \dif W_t - P_t H^\T R^{-1/2} \dif V_t
\end{equation}
for the \emph{filtering error} $E_t = X_t - \Xkf_t$ of the generic filter~\eqref{eq:filter}. This equation is derived by differentiating $E_t$, inserting the formulae for $\dif X_t$, $\dif Y_t$, and $\dif \Xkf_t$ from \eqref{eq:model} and \eqref{eq:filter} into the resulting stochastic differential equation, and recalling that $S = H^\T R^{-1} H$. The proof is given in \Cref{appendix:ConcProofs}. Observe that in the following $f - P_t S$ stands for the function $x \mapsto f(x) - P_t S x$.

\begin{theorem}\label{thm:ConInEq} Consider the generic filter~\eqref{eq:filter} for the continuous-time model~\eqref{eq:model} and let $\Lin_{x,P}$ satisfy \Cref{ass:Lt}. Suppose that there are positive constants $\lambda_P$ and $\lambda$ and time $T \geq 0$ such that $\sup_{t \geq 0} \trace(P_t) \leq \lambda_P$ and
\begin{equation}\label{eq:contractivityCondition}
M(f - P_t S) = \sup_{x \in \R^{d_x}} \mu\big[ J_f(x) - P_t S \big] \leq - \lambda < 0
\end{equation}
holds for every $t \geq T$ almost surely. Denote $\beta(\delta) = \neper ( \sqrt{2 \delta} + \delta )$.
Then there are non-negative constants $C_\lambda$ (continuously dependent on $\lambda$, $M(f)$, $N(f)$, $\trace(S)$, and $\lambda_P$) and $C_T$ such that, for any $t \geq T$ and $\delta > 0$, we have the \textbf{exponential concentration inequality}
\begin{equation}\label{eq:MainConcIneq}
  \mathbb{P}\Bigg[ \norm[0]{E_t}^2 \geq \bigg( C_T \neper^{-2\lambda(t-T)} + \frac{\trace(Q) + 2C_\lambda \lambda_P + \trace(S)\lambda_P^2}{2\lambda}\bigg) \beta(\delta) \Bigg] \leq \neper^{-\delta}
\end{equation}
and the \textbf{mean square filtering error bound}
\begin{equation}\label{eq:MainMeanSquare}
\expec\big( \norm[0]{E_t}^2 \big) \leq \expec(\norm[0]{E_T}^2) \neper^{-2\lambda(t-T)} + \frac{\trace(Q) + 2C_\lambda \lambda_P + \trace(S)\lambda_P^2}{2\lambda}.
\end{equation}
\end{theorem}

Several aspects of this theorem and its assumptions are discussed next.

\paragraph{\textbf{Assumption~\eqref{eq:contractivityCondition}}}

The assumption
\begin{equation*}
\sup_{t \geq T} M(f - P_t S) = \sup_{t \geq 0} \, \sup_{x \in \R^{d_x}} \mu\big[ J_f(x) - P_t S \big] < 0
\end{equation*}
is a time-uniform condition on contractivity of the filtering error process $E_t$. Indeed, it is the uniformity of this condition that enables the proof of \Cref{thm:ConInEq}. We are essentially ignoring any non-linear couplings between elements of $X_t$ that would need to be exploited were the analysis to be significantly extended and improved; see~\cite[Section 4]{DelMoralTugaut2018} for more discussion.
Even if one were to ignore issues with uniformity, the condition is still an extremely stringent one as it does not necessarily hold even for stable Kalman--Bucy filters for linear time-invariant systems. The Kalman--Bucy filter for the linear system
\begin{align*}
\dif X_t &= A X_t \dif t + Q^{1/2} \dif W_t, \\
\dif Y_t &= H X_t \dif t + R^{1/2} \dif V_t
\end{align*}
is
\begin{align*}
\dif \Xkf_t &= A \Xkf_t \dif t + P_t H^\T R^{-1} \big( \dif Y_t - H \Xkf_t \dif t \big), \\
\partial_t P_t &= A P_t + P_t A^\T + Q - P_t S P_t.
\end{align*}
Under certain observability and stabilizability conditions~\cite{Wonham1968,DeNicolaoGevers1992,CallierWinkin1995} the error covariance has a limiting steady state: $P_t \to P$ as $t \to \infty$ for the solution $P$ of the \emph{algebraic Riccati equation} $A P + P A^\T + Q - P S P = 0$.
Furthermore, the system $\partial_t x_t = (A - P S) x_t$ (i.e., homogeneous part of the linear filtering error equation) is exponentially stable in the usual sense that the eigenvalues of the system matrix are located in the left half-plane: $\alpha(A - P S) \coloneqq \max_{i=1,\ldots,d_x} \mathrm{Re} \big[ \lambda_i(A - P S) \big] < 0$.
However, the general inequality linking $\alpha(A-PS)$ and \sloppy{${M(A-PS) = \mu(A-PS)}$} is in the ``wrong'' direction~\cite[Equation (1.3)]{Soderlind2006}: $\alpha(A - PS) \leq \mu(A-PS)$.
That is, assumption~\eqref{eq:contractivityCondition} need not be satisfied even by stable filters for linear systems.
However, it often occurs that stability or forgetting theorems for non-linear filters do not completely cover the linear case; see for instance the results in~\cite{AtarZeitouni1997,BudhirajaOcone1999}.

\paragraph{\textbf{Error covariance}}
As here, the uniform boundedness of $P_t$ is assumed in almost every article on the stability of non-linear Kalman filters (though we discuss in \Cref{sec:examples} how to verify this assumption).
  What is less explicit is that in many cases the assumption~\eqref{eq:contractivityCondition} enforces a lower bound on $P_t$ because ``negativity'' of the term $-P_t S$ may be needed to ensure that $M(f-P_t S) < 0$.
  This behavior is discussed in more detail in \Cref{sec:covInf} in the context of covariance inflation.
  In literature it is in fact often explicitly assumed that the smallest eigenvalue of $P_t$ remains bounded away from zero (e.g.,~\cite{ReifGuntherYazUnbehauen2000,XiongZhangChan2006,XiongLiuZhang2009,KlugeReifBrokate2010,LiXia2012}).

\paragraph{\textbf{Constant $\pmb{C_\lambda}$}} For the EKF, the constant $C_\lambda$ is zero. For Gaussian assumed density and integration filters it was shown in \Cref{sec:slf,sec:ukf} that $C_g = M(g) - N(g)$. Because $M(f-P_t S) \leq -\lambda$ and $N(f-P_t S) \geq N(f) - \trace(S) \lambda_P$, these filters have $C_\lambda = -\lambda - N(f) + \trace(S) \lambda_P$.

\paragraph{\textbf{Dimensional dependency}}
  The error bounds of \Cref{thm:ConInEq} are strongly dependent on the dimensionality of the state space, $d_x$.
  The dimensional dependency is most clearly manifested in the term $\trace(Q)$ which grows linearly in $d_x$ if the noise variances for different dimensions are of the same order.
  From the examples in \Cref{sec:examples} it is seen that other constants in the bounds behave similarly.
  For example, in the setting of \Cref{prop:ContDrift} we have $\trace(S) = s d_x$ for $s > 0$ and $\lambda_P$ is the sum of traces of two $d_x \times d_x$ matrices.
  Note that the implications for the actual estimation error remain unclear as the bounds of \Cref{thm:ConInEq} appear to be very conservative (see~\Cref{sec:numerics}).


\section{Discrete-Time Models and Filters}\label{sec:discrete}

This section analyzes discrete-time systems and filters. First, we introduce a class of generic discrete-time filters analogous to continuous filters defined in \Cref{sec:filters} and then provide a discrete-time analog of \Cref{thm:ConInEq}. When necessary, we differentiate between the continuous and discrete cases by reserving $k$ for discrete time-indices and using an additional subscript $d$ for parameters related to the discrete case.

\subsection{A Class of Discrete-Time Filters for Non-Linear Systems}

In discrete time, we consider systems of the form
\begin{subequations}\label{eq:modelDisc}
\begin{align}
X_k &= f(X_{k-1}) + Q^{1/2} W_k, \\
Y_k &= H X_k + R^{1/2} V_k,
\end{align}
\end{subequations}
where $W_k \in \R^{d_x}$ and $V_k \in \R^{d_y}$ are independent standard Gaussian random vectors. The drift $f$ is assumed to be Lipschitz (i.e., $\norm[0]{J_f} = \sup_{x \in \R^{d_x}} \norm[0]{J_f(x)} < \infty$).
We again consider a linear functional $\Lin_{x,P}$ satisfying the basic properties listed in \Cref{sec:filters}. However, \Cref{ass:Lt} needs to be replaced with a slightly modified version.

\begin{assumption}\label{ass:LtDisc} For any differentiable $g \colon \R^{d_x} \to \R^{d_x}$ with finite $\norm[0]{J_g}$ there is a constant $C_g \geq 0$, which varies continuously with $\norm[0]{J_g}$, such that
\begin{equation*}
\norm[1]{g(x) - \Lin_{{\tilde{x}},P}(g)}^2 \leq \norm[0]{J_g}^2 \norm[0]{x - {\tilde{x}}}^2 + C_g \trace(P)
\end{equation*}
for any points $x,{\tilde{x}} \in \R^{d_x}$ and any $P \in \R^{d_x \times d_x}$.
\end{assumption}

Again, this assumption says that $\Lin_{{\tilde{x}},P}(g)$ cannot deviate too much from $g({\tilde{x}})$ since the standard Lipschitz bound is $\norm[0]{g(x) - g({\tilde{x}})} \leq \norm[0]{J_g} \norm[0]{x - {\tilde{x}}}$. A generic discrete-time filter for the system~\eqref{eq:modelDisc} produces the state estimates
\begin{equation}\label{eq:DiscGenFilter}
\Xkf_k = \Lin_{\Xkf_{k-1},P_{k-1}} (f) + P_{k \mid k-1} H^\T \big( HP_{k \mid k-1} H^\T + R \big)^{-1} \big[ Y_k - H \Lin_{\Xkf_{k-1},P_{k-1}}(f) \big],
\end{equation}
where $P_k$ and $P_{k \mid k-1}$ are user-specified positive-definite $d_x \times d_x$ matrices allowed to depend on the state estimates and measurements up to time $k-1$. 

\subsection{Kalman Filters for Discrete-Time Non-Linear Systems}\label{sec:KFdisc}

Like Kalman--Bucy filters of \Cref{sec:kalmanBucy}, a Kalman filter for the discrete-time model~\eqref{eq:modelDisc} computes approximations $\Xkf_k$ and $P_k$ to the filtering means and covariances \sloppy{${\expec( X_k \mid Y_1,\ldots,Y_k)}$} and $\var(X_k \mid Y_1,\ldots,Y_k)$. Such a filter consists of the \emph{prediction step}
\begin{subequations}\label{eq:kalmanPrediction}
\begin{align}
\Xkf_{k \mid k-1} &= \Lin_{\Xkf_{k-1},P_{k-1}}(f), \\
P_{k \mid k-1} &= \Lric_{\Xkf_{k-1},P_{k-1}}(f) + \Qt, \label{eq:Pprediction}
\end{align}
\end{subequations}
where $\Lric_{x,P}$ maps functions to positive-semidefinite matrices and $\Qt$ is again a potentially tuned version of $Q$, and the \emph{update step}
\begin{subequations}\label{eq:kalmanUpdate}
\begin{align}
K_k &= P_{k \mid k-1} H^\T \big( HP_{k \mid k-1} H^\T + R \big)^{-1}, \\
\Xkf_k &= \Xkf_{k \mid k-1} + K_k \big( Y_k - H \Xkf_{k \mid k-1} \big), \\
P_k &= ( \Id - K_k H) P_{k \mid k-1}. \label{eq:Pupdate}
\end{align}
\end{subequations}
The matrices $K_k$ are discrete-time versions of the Kalman gain matrices in~\eqref{eq:definition-K-S}.
All standard extensions of the Kalman filter for non-linear systems fit this framework. For example, $\Lin_{x,P}(g) = \Lin_{x,P}^\textsm{EKF}(g)= g(x)$ and $\Lric_{x,P}(g) = \Lric_{x,P}^\textsm{EKF(d)}(g) = J_g(x) P J_g(x)^\T$ yield the extended Kalman filter while
\begin{align*}
\Lin_{x,P}(g) &= \Lin_{x,P}^\textsm{ADF}(g) = \int_{\R^{d_x}} g(z) \gauss(z \mid x, P) \dif z, \\
\Lric_{x,P}(g) &= \Lric_{x,P}^\textsm{ADF(d)}(g) = \int_{\R^{d_x}} \big[ g(z) - \Lin_{x,P}^\textsm{ADF}(g) \big] \big[g(z) - \Lin_{x,P}^\textsm{ADF}(g) \big]^\T \gauss(z \mid x, P) \dif z
\end{align*}
correspond to discrete-time Gaussian assumed density filters. Obviously, by replacing the exact integrals with their numerical approximations we obtain different discrete-time Gaussian integration filters. For the EKF, \Cref{ass:LtDisc} holds again with $C_g = 0$, whereas similar arguments as those appearing in Appendix~\ref{app:int-quad} show that $C_g = \norm[0]{J_g}$ for the Gaussian assumed density filter and Gaussian integration filters whose numerical integration rules satisfy the second-degree exactness condition~\eqref{eq:quadraticExact}.

\subsection{Stability of Discrete-Time Filters}\label{sec:ConcDisc}

Discrete-time stability analysis that follows is based on a non-linear difference equation for the filtering error $E_k = X_k - \Xkf_k$:
\begin{equation*}
\begin{split}
E_k ={}& f(X_{k-1}) + Q^{1/2} W_k - \Xkf_{k \mid k-1} - K_k \big( Y_k - H \Xkf_{k \mid k-1} \big) \\
={}& f(X_{k-1}) - \Lin_{\Xkf_{k-1},P_{k-1}}(f) - K_k H \big( X_k - \Xkf_{k \mid k-1} \big) + Q^{1/2} W_k - K_k R^{1/2} V_k \\
={}& ( \Id - K_k H ) \big[ f(X_{k-1}) - \Lin_{\Xkf_{k-1},P_{k-1}}(f) \big] + ( \Id - K_k H) Q^{1/2} W_k - K_k R^{1/2} V_k.
\end{split}
\end{equation*}
The full proof is similar to that of Theorem~\ref{thm:ConInEq} and is given in Appendix~\ref{appendix:ConcProofDisc}.

\begin{theorem}\label{thm:ConInEqDisc} Consider the generic discrete-time filter~\eqref{eq:DiscGenFilter} for the model~\eqref{eq:modelDisc} and let $\Lin_{x,P}$ satisfy \Cref{ass:LtDisc}. Suppose that there are positive constants $\lambda_P^p$, $\lambda_P^u$, and $\lambda_d$ such that $\sup_{k \geq 0} \trace(P_{k \mid k-1}) \leq \lambda_P^p$, $\sup_{k \geq 0} \trace(P_k) \leq \lambda_P^u$,
\begin{equation}\label{eq:contractivityConditionDisc}
\sup_{k \geq 1} \norm[0]{I - K_k H} \leq \lambda_d < \infty, \quad and \quad \sup_{k \geq 1} \norm[0]{J_f} \norm[0]{I - K_k H} \leq \lambda_{df} < 1
\end{equation}
hold almost surely. Denote $\beta(\delta) = \neper ( \sqrt{2 \delta} + \delta )$ and $\kappa = \sup_{k \geq 1} \norm[0]{K_k} \leq \lambda_P^p \norm[0]{H} \norm[0]{R^{-1}}$. Then there is a non-negative constant $C_f$ such that, for any $\delta > 0$, we have the \textbf{exponential concentration inequality}
\begin{equation}\label{eq:MainConcIneqDisc}
\begin{split}
\mathbb{P}\Bigg[ \norm[0]{E_k}^2 \geq \norm 4\beta(\delta) \bigg( &\lambda_{df}^k \Big[\norm[0]{\mu_0 - \hat{x}_0} + \norm[0]{\Sigma_0}^{1/2} \Big] \bigg. \Bigg. \\
&+ \Bigg. \bigg. \frac{\sqrt{\lambda_d^2 \trace(Q) + \kappa^2 \trace(R) } + \lambda_d (C_f \lambda_P^u)^{1/2} }{1 - \lambda_{df}} \bigg)^{2} \Bigg] \leq \neper^{-\delta}
\end{split}
\end{equation}
and the \textbf{mean square filtering error bound}
\begin{equation}\label{eq:MainMeanSquareDisc}
\expec\big( \norm[0]{E_k}^2 \big) \leq \lambda_{df}^{2k} \big( \norm[0]{\mu_0 - \hat{x}_0}^2 + \trace(\Sigma_0) \big) + \frac{\lambda_d^2 [ \trace(Q) + C_f \lambda_P^u ] + \kappa^2 \trace(R)}{1-\lambda_{df}^2}.
\end{equation}
\end{theorem}

\subsection{Accuracy of Measurements} \label{sec:meas-acc}

If the measurements are $Y_k = X_k + R^{1/2} V_k$, one can simply use them as state estimates. For certain regimes of the system parameters it can be shown that the mean square error bound of \Cref{thm:ConInEqDisc} is an improvement over that for such naive state estimators. Consider the discrete-time system~\eqref{eq:modelDisc} and suppose the measurement model is $Y_k = h X_k + \sqrt{r} V_k$ for some positive scalars $h$ and $r$. Error of the naive estimate $\Xkf_{Y,k} = Y_k / h$ is
\begin{equation*}
X_k - \Xkf_{Y,k} = X_k - Y_k/h = (\sqrt{r}/h) V_k,
\end{equation*}
which is a zero-mean Gaussian random vector with variance $r/h^2$. That is,
\begin{equation}\label{eq:MeasEstBound}
\expec\big( \norm[0]{X_k - \Xkf_{Y,k}}^2 \big) = d_y r / h^2. 
\end{equation}
If the assumptions of \Cref{thm:ConInEqDisc} hold, the mean square bound~\eqref{eq:MainMeanSquareDisc} is
\begin{equation}\label{eq:FilterBoundMeasEst}
\expec\big( \norm[0]{E_k}^2 \big) \leq \lambda_{df}^{2k} \trace(P_0) + \frac{\lambda_d^2 [ \trace(Q) + C_f \lambda_P^u ] + \kappa^2 d_y r }{1-\lambda_{df}^2},
\end{equation}
where $\lambda_{df} < 1$ and
\begin{equation*}
\kappa = \sup_{k \geq 1} \norm[0]{K_k} = \sup_{k \geq 1} \, \norm[1]{h P_{k \mid k-1}(h^2 P_{k \mid k-1} + r \Id)^{-1} } \leq \frac{h}{h^2 + r/\lambda_P^p} = \mathcal{O}(r^{-1})
\end{equation*}
as $r \to \infty$.
We observe that the bound~\eqref{eq:FilterBoundMeasEst} is smaller than~\eqref{eq:MeasEstBound} if $\trace(Q)$ and $C_f$ are sufficiently small and $r$ is sufficiently large. From \Cref{sec:KFdisc} we recall that $C_f = 0$ for the EKF and $C_f = \norm[0]{J_f}$ for the UKF and its relatives.
This result is intuitive: if there is little process noise but the measurement noise level is high, the filter is able to produce accurate estimates by following the dynamics. This also demonstrates that in the setting where \Cref{thm:ConInEqDisc} is applicable the bounds it yields are sensible.


\section{Example Models}\label{sec:examples}

This section examines three model classes for which certain Kalman filters satisfy \Cref{thm:ConInEq} or \ref{thm:ConInEqDisc}, possibly under sufficient covariance inflation. The models in \Cref{sec:contractive,sec:covInf} are \emph{fully observed}, by which we mean that $S = H^\T R^{-1} H = s \Id $ for some $s > 0$. This assumption, though admittedly strong, is commonly used in analysis of various non-linear filters~\cite{KellyLawStuart2014,DelMoralKurtzmannTugaut2016,DelMoralTugaut2018,DelMoralKurtzmannTugaut2017,deWiljesReichStannat2017}. The model in \Cref{sec:IntVel} is \emph{fully detected} in the sense that the unobserved component is exponentially stable.

\subsection{Contractive Dynamics} \label{sec:contractive}

Stability analysis in~\cite{DelMoralKurtzmannTugaut2016} was restricted to the extended Kalman--Bucy filter for fully observed models with a \emph{contractive} (or \emph{uniformly monotone}) drift: $M(f) < 0$. This section applies Theorem~\ref{thm:ConInEq} to such models. The main difference to~\cite{DelMoralKurtzmannTugaut2016} is that the class of filters the analysis applies to is significantly expanded.

Consider a generalized Kalman--Bucy filter of the form~\eqref{eq:kalmanbucy} and suppose that there is $\ell_c$ such that $M(f) \leq -\ell_c < 0$. This means that the homogeneous system $\partial x_t = f(x_t)$ is exponentially stable: $x_t \to c$ with an exponential rate as $t \to \infty$ for some $c \in \R^{d_x}$. Assume also that the matrix-valued operator $\Lric_{x,P}$ in the Riccati equation~\eqref{eq:riccati} satisfies
\begin{equation}\label{eq:LricAss}
\trace [ \Lric_{x,P}(f) ] \leq M(f) \trace(P).
\end{equation}
As shown in~\cite{Karvonen2018}, this assumption is natural and satisfied by all Kalman--Bucy filters discussed in \Cref{sec:kalmanBucy}. From this assumption it follows that
\begin{equation*}
\begin{split}
\partial_t \trace(P_t) &= \trace \big[ \Lric_{\Xkf_t,P_t}(f) + \Lric_{\Xkf_t,P_t}(f)^\T \big] + \trace(\Qt) - \trace(P_t S P_t) \\
&\leq \trace \big[ \Lric_{\Xkf_t,P_t}(f) + \Lric_{\Xkf_t,P_t}(f)^\T \big] + \trace(\Qt) \\
&\leq -2 \ell_c \trace(P_t) + \trace(\Qt).
\end{split}
\end{equation*}
Consequently, by Grönwall's inequality (see \Cref{sec:gronwall}), $\trace(P_t) \leq \lambda_{P,t} \leq \lambda_P$, where
\begin{equation*}
  \lambda_{P,t} = \neper^{-2\ell_c t} \trace(P_0) + \trace(\Qt)/(2\ell_c) \leq \lambda_P \coloneqq \trace(P_0) + \trace(\Qt)/(2\ell_c).
\end{equation*}
Furthermore, if the model is in addition fully observed,
\begin{equation*}
M(f - P_t S) \leq M(f) + s \mu(-P_t) \leq -\ell_c.
\end{equation*}
That is, the assumptions of \Cref{thm:ConInEq} are satisfied for this class of exponentially stable and fully observed models for any positive-definite $\Qt$.

\begin{proposition}\label{prop:ContDrift} Consider a generic Kalman--Bucy filter~\eqref{eq:kalmanbucy}, defined by $\Lin_{x,P}$ satisfying \Cref{ass:Lt}, for the continuous-time model~\eqref{eq:model}. Suppose that there is a positive $\ell_c$ such that $M(f) \leq - \ell_c < 0$, \sloppy{${S = H^\T R^{-1} H = s \Id}$} for some $s > 0$, and that \eqref{eq:LricAss} holds. Then \Cref{thm:ConInEq} holds with $T = 0$, $\lambda = \ell_c$, and \sloppy{${\lambda_P = \trace(P_0) + \trace(\Qt)/(2\ell_c)}$}.
\end{proposition}

In particular, under the assumptions of the above proposition and when using the time-dependent bound $\lambda_{P,t}$, the concentration inequality~\eqref{eq:MainConcIneq} for the EKF takes the form
\begin{equation*}
\begin{split}
 \norm[0]{E_t}^2 \geq{}& \bigg( \expec(\norm[0]{E_0}^2)\neper^{-2\ell_c t} + \frac{\trace(Q) + d_x s [\neper^{-2\ell_c t} \trace(P_0) + \trace(\Qt)/(2\ell_c)]^2}{2\ell_c} \bigg) \beta(\delta) \\
\xrightarrow[t \to \infty]{} {}& \, \bigg[ \trace(Q) + \frac{d_x s \trace(\Qt)^2}{4\ell_c^2} \bigg] \frac{1}{2\ell_c} \beta(\delta)
\end{split}
\end{equation*}
with probability smaller than $\neper^{-\delta}$. This is, up to some constants, identical to (1.13), the corresponding result in~\cite{DelMoralKurtzmannTugaut2016} (note that Del Moral et al.\ have $\Qt = Q$) which in our notation and as $t \to \infty$ is
\begin{equation*}
  \norm[0]{E_t}^2 \geq \bigg( \trace(Q) + \frac{ s \trace(Q)^2}{2\ell_c} \bigg) \frac{1}{2\ell_c^2} \omega(\delta) \quad \text{with} \quad \omega(\delta) = 4\sqrt{2} \neper^2 \bigg( \frac{1}{2} + \delta + \sqrt{\delta} \bigg).
\end{equation*}

\subsection{Covariance Inflation}\label{sec:covInf}

Intuitively, if the state is observed linearly and ``well enough'', artificial inflation of the error covariance matrix $P_t$ should make the filter more stable (or robust) since this results in less emphasis being placed on the state dynamics, which mitigates instability due to non-linearity of the drift. \emph{Covariance inflation}, sometimes known in engineering literature as \emph{robust tuning}, is an important topic in the study of ensemble Kalman filters~\cite{KellyLawStuart2014,TongMajdaKelly2016} and has been suggested also for stabilizing the discrete-time UKF~\cite{XiongZhangChan2006,XiongLiuZhang2009}. It allows for considering models whose drift is not necessarily contractive, which is a case not covered by current theory in~\cite{DelMoralKurtzmannTugaut2016}.

Suppose that $S = s \Id$ for some positive $s$. Then
\begin{equation*}
\sup_{x \in \R^{d_x}} \mu\big[ J_f(x) - P_t S \big] \leq M(f) + s\mu(-P_t) = M(f) - s \mineig(P_t),
\end{equation*}
and it is evident that for large enough $\mineig(P_t)$ this quantity becomes negative as required in \Cref{thm:ConInEq}. Specifically, $\mineig(P_t) \geq (M(f) + \lambda)/s$ is sufficient to ensure that $\sup_{x \in \R^{d_x}} \mu [ J_f(x) - P_t S ] \leq -\lambda$. This can be achieved using covariance inflation in Kalman--Bucy filters by choosing a large enough tuned dynamic noise covariance matrix $\Qt$. For simplicity, consider the extended Kalman--Bucy filter. The inversion formula $\partial_t P_t^{-1} = - P_t^{-1} ( \partial_t P_t ) P_t^{-1}$ yields the Riccati equation
\begin{equation*}
\partial_t P_t^{-1} = - P_t^{-1} J_f(\Xkf_t) - J_f(\Xkf_t)^\T P_t^{-1} + S - P_t^{-1} Q P_t^{-1}
\end{equation*}
for the inverse error covariance. By using arguments similar to those appearing in~\cite{Karvonen2018} we can prove that 
\begin{equation}\label{eq:TraceInverse}
\trace(P_t^{-1}) \leq \frac{ \sqrt{\mineig(\Qt) \maxeig(S) / d_x + N(f)^2} - N(f)}{\mineig(\Qt) / d_x} + \alpha_1 \neper^{-\beta_1 t}
\end{equation}
for some positive constants $\alpha_1$ and $\beta_2$ that depend on the system parameters.
Since $\trace(P_t^{-1}) = \sum_{i=1}^{d_x} \lambda_i(P_t)^{-1}$, \eqref{eq:TraceInverse} implies the eigenvalue bound
\begin{equation*}
\mineig(P_t) \geq \frac{1}{ \trace(P_t^{-1}) } \geq \frac{\mineig(\Qt) / d_x}{\sqrt{\mineig(\Qt) \maxeig(S) / d_x + N(f)^2 } - N(f) + \alpha_2 \neper^{-\beta_2 t} }
\end{equation*}
for some positive constants $\alpha_2$ and $\beta_2$. As this eigenvalue bound grows as square root of $\mineig(\Qt)$, the inequality \sloppy{${\mineig(P_t) \geq (M(f) + \lambda)/s}$} that induces the stability condition~\eqref{eq:contractivityCondition} is satisfied when $\mineig(\Qt)$ and $t$ are large enough.

\subsection{Integrated Velocity Models}\label{sec:IntVel}

Let $h \neq 0$, $a_2,q_1,q_2,r > 0$, and $a_1$ be constants and \sloppy{${g \colon \R \to \R}$} a continuously differentiable function such that
\begin{equation}\label{eq:GassIntVel}
N(g) = \inf_{x \in \R} g'(x) \geq \ell_g > 0
\end{equation}
for a constant $\ell_g$. Consider the \emph{integrated velocity model}
\begin{equation}\label{eq:velocityModel}
\begin{split}
\dif \begin{bmatrix} X_{t,1} \\ X_{t,2} \end{bmatrix} &= \begin{bmatrix} a_1 X_{t,1} + a_2 X_{t,2} \\ - g(X_{t,2}) \end{bmatrix} \dif t + \begin{bmatrix} q_1^{1/2} & 0 \\ 0 & q_2^{1/2} \end{bmatrix} \dif V_t, \\
\dif Y_t &= \begin{bmatrix} h & 0 \end{bmatrix} X_t \dif t + r^{1/2} \dif W_t
\end{split}
\end{equation}
for a two-dimensional state $X_t = (X_{t,1}, X_{t,2}) \in \R^2$ of which one-dimensional measurements $Y_t$ are obtained. When $a_1 = 0$, the state component $X_{t,1}$ can be interpreted as the position of a target, obtained by integrating its velocity $X_{t,2}$. Using only position measurements we then try to estimate both the position and the velocity.

We now show the extended Kalman--Bucy filter~\eqref{eq:ekf} for this model satisfies \Cref{thm:ConInEq} if there is sufficient covariance inflation.
Because
\begin{equation*}
J_f(x) = \begin{bmatrix} a_1 & a_2 \\ 0 & -g'(x_2) \end{bmatrix} \hspace{0.3cm} \text{and} \hspace{0.3cm} S = \begin{bmatrix} h^2/r & 0 \\ 0 & 0 \end{bmatrix},
\end{equation*}
the EKF for the model~\eqref{eq:velocityModel} takes the form
\begin{align*}
\dif \Xkf_t ={}& \begin{bmatrix} a_1 \Xkf_{t,1} + a_2 \Xkf_{t,2} \\ -g(\Xkf_{t,2}) \end{bmatrix} \dif t + \begin{bmatrix} P_{t,11} & P_{t,12} \\ P_{t,12} & P_{t,22} \end{bmatrix} \begin{bmatrix} h/r \\ 0 \end{bmatrix} \begin{bmatrix} \dif Y_t - h\big(a_1 \Xkf_{t,1} + a_2 \Xkf_{t,2}\big) \dif t \end{bmatrix}, \\
\partial_t P_t ={}& \begin{bmatrix} a_1 & a_2 \\ 0 & -g'(\Xkf_{t,2}) \end{bmatrix} \begin{bmatrix} P_{t,11} & P_{t,12} \\ P_{t,12} & P_{t,22} \end{bmatrix} + \begin{bmatrix} P_{t,11} & P_{t,12} \\ P_{t,12} & P_{t,22} \end{bmatrix} \begin{bmatrix} a_1 & 0 \\ a_2 & -g'(\Xkf_{t,2}) \end{bmatrix} \\
&+ \begin{bmatrix} q_{\text{tu},1} & 0 \\ 0 & q_{\text{tu},2} \end{bmatrix} - \begin{bmatrix} P_{t,11} & P_{t,12} \\ P_{t,12} & P_{t,22} \end{bmatrix} \begin{bmatrix} h^2/r & 0 \\ 0 & 0 \end{bmatrix} \begin{bmatrix} P_{t,11} & P_{t,12} \\ P_{t,12} & P_{t,22} \end{bmatrix},
\end{align*}
where $q_{\text{tu},1}$ and $q_{\text{tu},2}$ are elements of the diagonal tuned noise covariance $\Qt$. Differential equations for the three distinct elements of $P_{t,11}$ are
\begin{align*}
\partial_t P_{t,11} &= 2a_1 P_{t,11} + q_{\text{tu},1} - s P_{t,11}^2 + 2a_2 P_{t,12},\\
\partial_t P_{t,12} &= \big[ a_1 - g'(\Xkf_{t,2}) - s P_{t,11} \big] P_{t,12} + a_2 P_{t,22},\\
\partial_t P_{t,22} &= -2g'(\Xkf_{t,2}) P_{t,22} + q_{\text{tu},2} - s P_{t,12}^2.
\end{align*}
From \eqref{eq:GassIntVel} it follows that $\partial_t P_{t,22} \leq -2\ell_g P_{t,22} + q_{\text{tu},2}$, which yields the upper bound
$P_{t,22} \leq \neper^{-2\ell_g t} P_{0,22} + q_{\text{tu},2}/(2\ell_g) \eqqcolon C_{22}(t)$.
Suppose that $P_{0,12} \geq 0$. Since $a_2 P_{t,22} > 0$, we have $P_{t,12} \geq 0$ for every $t \geq 0$. Thus $\partial_t P_{t,11} \geq 2a_1 P_{t,11} + q_{\text{tu},1} - s P_{t,11}^2$, and from this it can be established that~\cite[Lemma 3]{Karvonen2018}
\begin{equation*}\label{eq:P11Lower}
P_{t,11} \geq \frac{a_1 + ( s q_{\text{tu},1} + a_1^2 )^{1/2}  }{s} - \alpha_{1} \neper^{-\beta_{1} t}
\end{equation*}
for some positive constants $\alpha_{1}$ and $\beta_{1}$. It follows that
\begin{equation*}
a_1 - g'(x) - s P_{t,11} \leq a_1 - \ell_g - s P_{t,11} \leq a_1 - \ell_g - ( s q_{\text{tu},1} + a_1^2 )^{1/2} + s \alpha_{1} \neper^{-\beta_{1} t}.
\end{equation*}
That is, for every $0 < \lambda_{12} < \ell_g + ( s q_{\text{tu},1} + a_1^2 )^{1/2}$ there are $q_{\text{tu},1}$ and a time-horizon $T_{\lambda_{12}}$ such that
\begin{equation}\label{eq:lambda12cond}
a_1 - g'(x) - s P_{t,11} \leq a_1 - \ell_g - s P_{t,11} \leq -\lambda_{12} < 0
\end{equation}
when $t \geq T_{\lambda_{12}}$. Thus $\partial_t P_{t,12} \leq -\lambda_{12} P_{t,12} + a_2 P_{t,22} \leq -\lambda_{12} P_{t,12} + a_2 C_{22}(t)$ for $t \geq T_{\lambda_{12}}$, implying that there is a time-uniform upper bound $C_{12}$ on $P_{t,12}$. From this we obtain an upper bound for $P_{t,11}$:
\begin{equation*}
\partial_t P_{t,11} = 2a_1 P_{t,11} + q_{\text{tu},1} - s P_{t,11}^2 + 2a_2 P_{t,12} \leq 2a_1 P_{t,11} - s P_{t,11}^2 + q_{\text{tu},1} + 2a_2 C_{12}
\end{equation*}
implies that
\begin{equation*}
P_{t,11} \leq \frac{a_1 + ( s(q_{\text{tu},1} + 2a_2 C_{12}) + a_1^2)^{1/2} }{s} + \alpha_2 \neper^{-\beta_2 t}
\end{equation*}
for some positive constants $\alpha_{2}$ and $\beta_{2}$. Since the both diagonal elements $P_{t,11}$ and $P_{t,22}$ are bounded, we have thus obtained an upper bound on $\trace(P_t)$.

Finally, to show that \Cref{thm:ConInEq} is applicable, we need to prove that the matrix
\begin{equation*}
A = \big( J_f(x) - P_t S \big)_{\text{sym}} = \begin{bmatrix} a_1 - s P_{t,11} & a_2 \\ -s P_{t,12} & -g'(x)\end{bmatrix}_\text{sym} = \begin{bmatrix} 2(a_1 - s P_{t,11}) & a_2 - s P_{t,12} \\ a_2 - s P_{t,12} & -2 g'(x) \end{bmatrix}
\end{equation*}
is negative-definite for every $x \in \R$ and large enough $t$.
Eigenvalues of this matrix are
\begin{equation*}
\frac{1}{2} \Big( \trace(A) \pm \sqrt{\trace(A)^2 - 4 \det(A)} \Big).
\end{equation*}
Having previously selected $q_{\text{tu},1}$ and $T_{\lambda_{12}}$ such that
\begin{equation*}
\frac{1}{2} \trace(A) = a_1 - g'(x) - s P_{t,11} \leq -\lambda_{12} < 0,
\end{equation*}
we see that the larger of the eigenvalues is negative since \sloppy{${\sqrt{\trace(A)^2 - 4\det(A)} < \abs[0]{\trace(A)}}$}.
We have thus proved that error covariance inflation can be used to induce provable stability of the extended Kalman--Bucy filter for the integrated velocity model~\eqref{eq:velocityModel}.


\section{Numerical Examples}\label{sec:numerics}

This section contains numerical examples that validate the mean square error bound of \Cref{thm:ConInEq} for the extended and unscented Kalman--Bucy filters applied to two toy models.

\subsection{Contractive Dynamics}

In this example we consider the EKF and the UKF for the fully observed model
\begin{equation}\label{eq:contractiveModel}
\begin{split}
\dif X_t &= f(X_t) \dif t + \dif W_t, \\
\dif Y_t &= X_t \dif t + \sqrt{8} \dif V_t,
\end{split}
\end{equation}
that is initialised from $X_0 \sim \gauss(0,10^{-2})$ and is specified by the drift function
\begin{equation*}
f(x) = \begin{bmatrix} -x_3 \Big( 1 + \frac{1}{1+x_3^2} \Big) - 3 x_1 \\ -x_1 - x_2 - x_3 \\ x_1^2 \neper^{-x_1^2-x_3^2} - x_1 - 2 x_3 \end{bmatrix}.
\end{equation*}
We compute $N(f) \approx -4.5046$ and $ M(f) \approx -0.5947$.
Hence the model is exponentially stable and the assumptions of \Cref{prop:ContDrift} are satisfied with $\ell_c = -M(f)$. For a generic Kalman--Bucy filter, this proposition yields the bound (when $\Qt = Q$)
\begin{equation*}
\trace(P_t) \leq \lambda_P = \trace(P_0) + \trace(Q)/(2 \ell_c) \approx 2.552.
\end{equation*}
We use the initialization $\Xkf = \expec(X_0)= 0$. The mean square bound of \Cref{thm:ConInEq} is
\begin{equation}\label{eq:ex1Bound}
\expec\big( \norm[0]{X_t - \Xkf_t}^2 \big) \leq \trace(P_0) \neper^{-2\lambda t} + \frac{\trace(Q) + 2 C_\lambda \lambda_P + \trace(S) \lambda_P^2 }{2\ell_c},
\end{equation}
where $C_\lambda = 0$ for the EKF and 
$C_\lambda = M(f) - N(f) + \trace(S) \lambda_P \approx  4.867$
for the UKF (see \Cref{sec:stability}). Note that this is merely a shortcoming of the proof technique we have used rather than a manifestation of greater accuracy of the EKF.

\Cref{fig:contractive} depicts the theoretical upper bounds on $\expec(\norm[0]{X_t - \Xkf_t}^2)$ for the EKF and the UKF and the empirical mean square error based on 1,000 state and measurement trajectory realizations. The results were obtained using Euler--Maruyama discretization with step-size 0.01. It is evident that the theoretical bounds are valid and somewhat conservative, which is quite typical in stability theory of non-linear Kalman filters (see, e.g., numerical examples in~\cite{ReifGuntherYazUnbehauen1999,ReifGuntherYazUnbehauen2000}).

\begin{figure}[t!]
  \centering
  \includegraphics[width=\columnwidth]{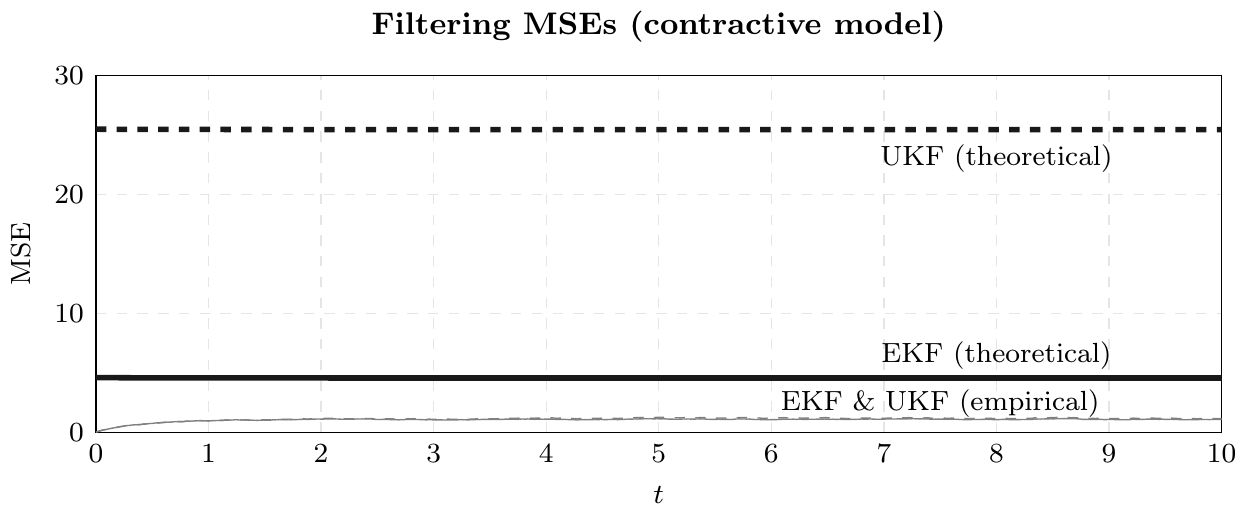}
  \caption{Empirical mean square filtering errors based on 1,000 state and measurement trajectory realisations and the theoretical error bounds~\eqref{eq:ex1Bound} for the EKF and the UKF applied to the model~\eqref{eq:contractiveModel}. Time-averaged empirical MSEs are $1.058$ (EKF) and $1.128$ (UKF).}\label{fig:contractive}
\end{figure}

\subsection{Integrated Velocity Model}

We now validate the theoretical bounds obtained in \Cref{sec:stability} on the integrated velocity model discussed in \Cref{sec:IntVel}.
Consider the EKF for the integrated velocity model~\eqref{eq:velocityModel} with the parameters $a_1 = 0$, $a_2 = 1$, $q_1 = q_2 = 0.05$, $h = 1$, $r = 0.05$, $\Xkf_0 = 0$, $\mu_0 = 0$, $P_0 = 0.01 \Id$, and
\begin{equation*}
g(x) = x \bigg( 1+\frac{\sin x}{1+x^2} \bigg), \quad g'(x) = 1 + \frac{(x^3+x)\cos x - (x^2-1)\sin x }{(1+x)^2}.
\end{equation*}
The maximum and minimum of $g'$ are $\sup_{x \in \R} g'(x) \approx 1.581$ and $\inf_{x \in \R} g'(x) \approx 0.419$.
That is, $g$ satisfies \eqref{eq:GassIntVel} with $\ell_g = 0.419$. Based on the derivations in \Cref{sec:IntVel} we are able to compute that $\trace(P_t) \leq \lambda_P \approx 0.173$ for all sufficiently large $t$. Because $a_1 = 0$, no covariance inflation is needed for \eqref{eq:lambda12cond} to hold.
In this particular case, the value $\lambda = 0.5478$ can be used in \Cref{thm:ConInEq}.

\Cref{fig:intvel} depicts the limiting (i.e., all exponentially decaying terms are disregarded) theoretical mean square filtering error bound for the EKF thus obtained and the empirical mean square error based on 1,000 state and measurement trajectory realisations. Again, Euler--Maruyama discretization with step-size 0.01 was used.

\begin{figure}[t!]
  \centering
  \includegraphics[width=\columnwidth]{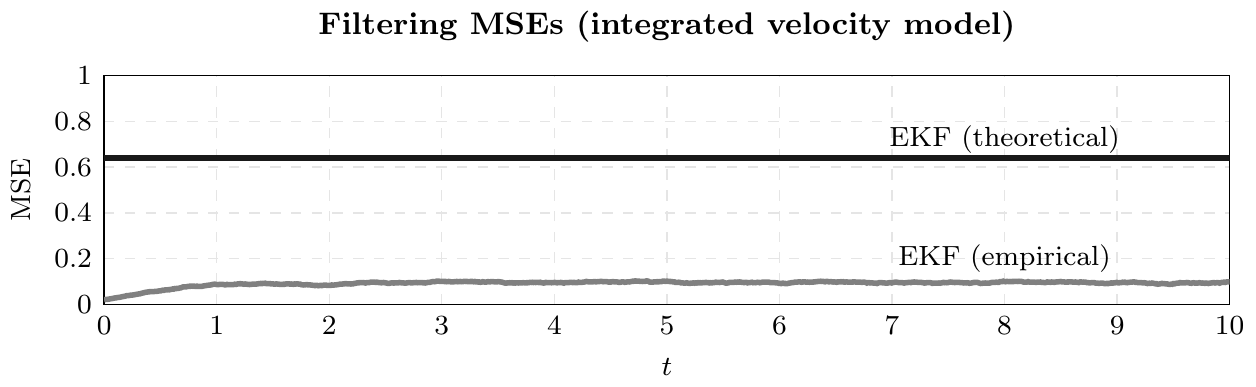}
  \caption{Empirical mean square filtering errors based on 1,000 state and measurement trajectory realisations and the limiting theoretical error bound for the EKF applied to an integrated velocity model.}\label{fig:intvel}
\end{figure}


\section{Conclusions and Discussion}

In this article we have shown that large classes of generic filters for both continuous and discrete-time systems with non-linear state dynamics and linear measurements are stable in the sense of time-uniformly bounded mean square filtering error if certain stringent conditions on boundedness of error covariance matrices and the filtering error process are met. The analysis extends the previous work~\cite{DelMoralKurtzmannTugaut2016} for the extended Kalman--Bucy filter and fully observed and exponentially stable state models. Our main contributions have been generalizations to models that need not be fully observed or exponentially stable and to a large class of commonly used extensions of the Kalman--Bucy or Kalman filter to non-linear systems, such as Gaussian assumed density filters and their numerical approximation, including the unscented Kalman filter. In \Cref{sec:examples}, we have also presented three different classes of models and filters that satisfy the stability assumptions. This is in stark contrast to earlier work for, for example, the UKF that has relied on unverifiable assumptions on certain auxiliary random matrices~\cite{XiongZhangChan2006}.

The results rely on admittedly very strong conditions on the filtering error process. These conditions cannot be significantly relaxed unless a more sophisticated proof technique is devised as the technique we have used essentially neglects potential non-linear couplings of state components. It appears to us that no such technique exists at the moment. The only non-trivial and interesting extensions that we believe are possible are to fully detected systems, essentially generalizations of the integrated velocity model we considered in \Cref{sec:IntVel}, where not all state components need to be (fully) observed, but those that are not must be exponentially stable so that their effect on observed components is small.


\appendix

\section{Gaussian Assumed Density and Integration Filters} \label{app:int-quad}

This appendix proves that the Gaussian assumed density and integration filters defined in Sections~\ref{sec:slf} and~\ref{sec:ukf} satisfy Assumption~\ref{ass:Lt}.

For the Gaussian assumed density filter the functional $\Lin_{x,P}^\textsm{ADF}$ is defined in~\eqref{eq:L-slf}. For any differentiable $g \colon \R^{d_x} \to \R^{d_x}$ we have
\begin{equation*}
\begin{split}
\inprod[\big]{x-{\tilde{x}}}{g(x) - \Lin_{{\tilde{x}},P}^\textsm{ADF}(g)} \hspace{-2.5cm}& \\
&= \inprod[\big]{x-{\tilde{x}}}{g(x) - \expec_{\gauss({\tilde{x}},P)}(g)} \\
&= \int_{\R^{d_x}} \inprod[\big]{(x-z) + (z-{\tilde{x}})}{g(x) - g(z)} \, \gauss(z \mid {\tilde{x}}, P) \dif z \\
&= \int_{\R^{d_x}} \inprod[\big]{x-z}{g(x) - g(z)} \, \gauss(z \mid {\tilde{x}}, P) \dif z - \int_{\R^{d_x}} \inprod[\big]{z - {\tilde{x}}}{g(z)} \, \gauss(z \mid {\tilde{x}}, P) \dif z.
\end{split}
\end{equation*}
The first term can be bounded as
\begin{equation*}
\begin{split}
\int_{\R^{d_x}} \inprod[\big]{x-z}{g(x) - g(z)} \, \gauss(z \mid {\tilde{x}}, P) \dif z \hspace{-2.5cm}&\\
&\leq M(g) \int_{\R^{d_x}} \norm[0]{x-z}^2 \, \gauss(z \mid {\tilde{x}}, P) \dif z \\
&= M(g) \bigg( \int_{\R^{d_x}} \big( \norm[0]{z-{\tilde{x}}}^2 + \norm[0]{x-{\tilde{x}}}^2 \big) \, \gauss(z \mid {\tilde{x}}, P) \dif z \bigg) \\
&= M(g) \big[ \norm[0]{x-{\tilde{x}}}^2 + \trace(P) \big],
\end{split}
\end{equation*}
whereas the second has the bound
\begin{equation*}
\begin{split}
- \int_{\R^{d_x}} \inprod[\big]{z - {\tilde{x}}}{g(z)} \, \gauss(z \mid {\tilde{x}}, P) \dif z &= - \int_{\R^{d_x}} \inprod[\big]{z - {\tilde{x}}}{g(z) - g({\tilde{x}})} \, \gauss(z \mid {\tilde{x}}, P) \dif z \\
&\leq - N(g) \int_{\R^{d_x}} \norm[0]{z-{\tilde{x}}}^2 \, \gauss(z \mid {\tilde{x}}, P) \dif z \\
&= -N(g) \trace(P).
\end{split}
\end{equation*}
These estimates show that \Cref{ass:Lt} holds with $C_g = M(g) - N(g) \geq 0$.

For the Gaussian integration filter the functional $\Lin_{x,P}^\textsm{int}$ is defined in~\eqref{eq:cubature}. That~\eqref{eq:quadraticExact} holds for any polynomial of total degree up to two implies that $\sum_{i=1}^n w_i = 1$ and $\sum_{i=1}^n w_i \sqrt{P}\xi_i = 0$ since $\Lin_{x,P}^\textsm{int}(1) = \Lin_{\gauss(x,P)}(1) = 1$ and $\Lin_{x,P}^\textsm{int}(p) = \expec_{\gauss(x,P)}(p) = 0$ for $p(z) = z - x$. Under these assumptions we can proceed as above:
\begin{align*}
\inprod[\big]{x - \tilde{x}}{g(x) - \Lin_{\tilde{x},P}^\textsm{int}(g)} ={}& \inprod[\bigg]{x - \tilde{x}}{g(x) - \sum_{i=1}^n w_i g\big(\tilde{x} + \sqrt{P} \xi_i\big) } \\
={}& \sum_{i=1}^n w_i \inprod[\Big]{x - \big(\tilde{x} + \sqrt{P} \xi_i\big) + \sqrt{P} \xi_i }{ g(x) -  g\big(\tilde{x} + \sqrt{P} \xi_i\big) } \\
={}& \sum_{i=1}^n w_i \bigg( \inprod[\Big]{x - \big(\tilde{x} + \sqrt{P} \xi_i\big)}{ g(x) -  g\big(\tilde{x} + \sqrt{P} \xi_i\big) } \\
&\quad\quad\quad\quad+ \inprod[\Big]{\sqrt{P} \xi_i}{ g(x) -  g\big(\tilde{x} + \sqrt{P} \xi_i\big) } \bigg)
\end{align*}
Hence
\begin{equation*}
\begin{split}
\inprod[\big]{x - \tilde{x}}{g(x) - \Lin_{\tilde{x}y,P}(g)} \leq{}&  M(g) \sum_{i=1}^n w_i \norm[1]{x-\big(\tilde{x}+\sqrt{P}\xi_i\big) }^2 \\
&+ \sum_{i=1}^n w_i \inprod[\Big]{\sqrt{P} \xi_i}{ g(x) -  g\big(\tilde{x} + \sqrt{P} \xi_i\big) }.
\end{split}
\end{equation*}
The first term is a sigma-point approximation of a quadratic function. Using \eqref{eq:quadraticExact} and proceeding again as in \Cref{sec:slf},
\begin{equation*}
\sum_{i=1}^n w_i \norm[1]{x-\big({\tilde{x}}+\sqrt{P}\xi_i\big) }^2 = \int_{\R^{d_x}} \norm[0]{x-z}^2 \, \gauss(z \mid {\tilde{x}}, P) \dif z = \norm[0]{x-{\tilde{x}}}^2 + \trace(P).
\end{equation*}
To bound the second term, notice that
\begin{equation*}
\begin{split}
  \sum_{i=1}^n w_i &\inprod[\Big]{\sqrt{P} \xi_i}{ g(x) -  g\big({\tilde{x}} + \sqrt{P} \xi_i\big) } \\
  &= -\sum_{i=1}^n w_i \inprod[\Big]{{\tilde{x}} - \big({\tilde{x}} + \sqrt{P} \xi_i\big) }{ g({\tilde{x}}) -  g\big({\tilde{x}} + \sqrt{P} \xi_i\big) } \\
&\leq -N(g) \sum_{i=1}^n \norm[1]{\sqrt{P} \xi_i }^2 \\
&= -N(g) \trace(P)
\end{split}
\end{equation*}
by exactness of $\Lin_{\tilde{x},P}^\textsm{int}$ for quadratic polynomials. \Cref{ass:Lt} thus holds with the constant $C_g = M(g) - N(g)$.

\section{Grönwall's Inequalities} \label{sec:gronwall}

Classical Grönwall inequalities are a basic ingredients in our proofs.
\paragraph{\textbf{Continuous version}} Suppose that $\beta_t$ is a continuous real-valued function of $t \in \R$ and  $x_t$ is continuously differentiable on $\R_+$ and satisfies $\partial_t x_t \leq \alpha x_t + \beta_t$, $t \geq 0$, for some constant $\alpha$. Then Grönwall's inequality states that
\begin{equation}\label{eq:Gronwall1}
x_t \leq x_0 \neper^{\alpha t} + \int_0^t \neper^{\alpha(t-s)} \beta_s \dif s
\end{equation}
for every $t \geq 0$. If $\beta_t \equiv \beta$, \eqref{eq:Gronwall1} reduces to
\begin{equation}\label{eq:Gronwall2}
x_t \leq x_0 \neper^{\alpha t} - (1-\neper^{\alpha t})\beta / \alpha. 
\end{equation}
The form of \eqref{eq:Gronwall2} that we need the most is the one where $\beta_t \equiv \beta \geq 0$ and $\alpha = -\gamma$ for $\gamma > 0$. Then the inequality takes the form $x_t \leq x_0 \neper^{-\gamma t} + \beta/\gamma$.

\paragraph{\textbf{Discrete version}} Let \sloppy{${0 \leq \alpha < 1}$} and $\beta \geq 0$ and suppose that $x_k \geq 0$ satisfy the difference inequality $x_k \leq \alpha x_{k-1} + \beta$ for $k \geq 1$. Then the discrete Grönwall's inequality states that
\begin{equation} \label{eq:Gronwall-discrete}
  x_k \leq \alpha^k x_0 + \beta \sum_{n=0}^{k-1} \alpha^n \leq \alpha^k x_0 + \frac{\beta}{1 - \alpha}.
\end{equation}

\section{Bernstein's Concentration Inequality}\label{appendix:bernstein}

In contrast to Del Moral et al.\@~\cite{DelMoralKurtzmannTugaut2016} who base their exponential concentration inequality for the EKF on the concentration inequality appearing in Proposition 11.6.6 of~\cite{DelMoral2013}, we use a version of the classical Bernstein inequality.

\begin{theorem}[Bernstein's inequality]\label{thm:bernstein} Let $X$ be a non-negative random variable. Suppose that there is $\alpha > 0$ such that
$\expec(X^n) \leq n^n \alpha^n$
for every integer $n \geq 2$. Then
\begin{equation}\label{eq:ConcBernstein}
\mathbb{P} \big[ X \geq \alpha \neper \big( \sqrt{2 \delta} + \delta \big) \big] \leq \neper^{-\delta}
\end{equation}
for any $\delta > 0$.
\end{theorem}
\begin{proof} By the standard Stirling bound,
\begin{equation*}
\expec(X^n) \leq n^n \alpha^n \leq \frac{n!}{\sqrt{2\pi}} \neper^n \alpha^n \leq \frac{n!}{2} ( \neper \alpha)^n
\end{equation*}
for every $n \geq 2$. The ``standard'' version of Bernstein's inequality~(e.g.,~\cite[Theorem 2.10]{BoucheronLugosiMassart2013}) posits that $\expec(X^2) \leq \sigma < \infty$ and $\expec(X^n) \leq \frac{n!}{2} \sigma \gamma^{n-2}$ for some $\sigma > 0$ and $\gamma > 0$ and every $n \geq 3$ imply $\mathbb{P} [ X \geq \sqrt{2 \sigma \delta} + \gamma \delta ] \leq \neper^{-\delta}$ for any $\delta > 0$. Thus, setting $\gamma = \neper \alpha$ and $\sigma = \gamma^2$ produces the claim.
\end{proof}

The concentration inequality used in \cite[Equation~(3.6)]{DelMoralKurtzmannTugaut2016} is based on the same moment bounds $\expec(X^n) \leq n^n \alpha^n$ but states instead that
\begin{equation*}
\mathbb{P} \, \bigg[ X \geq \frac{\alpha \neper^2}{\sqrt{2}} \bigg( \frac{1}{2} + \sqrt{\delta} + \delta \bigg) \bigg] \leq \neper^{-\delta}
\end{equation*}
for any $\delta > 0$. Since
\begin{equation*}
\mathbb{P} \big[ X \geq \alpha \neper \big( \sqrt{2 \delta} + \delta \big) \big] = \mathbb{P} \bigg[ X \geq \alpha \sqrt{2} \neper \bigg( \sqrt{\delta} + \frac{1}{\sqrt{2}}\delta \bigg) \bigg]
\end{equation*}
and $\sqrt{2}\neper < \neper^2/\sqrt{2}$, the inequality~\eqref{eq:ConcBernstein} is the tighter of the two for every $\delta > 0$.

\begin{lemma}\label{lemma:ChiSquareLemma} Let $X \sim \gauss(m, P)$ be a $d$-dimensional Gaussian random vector with a positive-semidefinite covariance $P$. Then
\begin{equation*}
\expec(\norm[0]{X}^{2n})^{1/n} \leq 4\big( \norm[0]{m}^{2} + \norm[0]{P} [d+2] n \big)
\end{equation*}
for every $n \in \mathbb{N}$. If $m = 0$, the inequality is 
\begin{equation*}
  \expec(\norm[0]{X}^{2n})^{1/n} \leq \norm[0]{P} (d+2) n.
\end{equation*}
\end{lemma}

\begin{proof} We know that $X = m + P^{1/2} U$ for a standard normal $U \in \R^d$. Therefore
\begin{equation*}
\begin{split}
\expec(\norm[0]{X}^{2n}) = \expec\big( \big[\norm[0]{m} + \norm[0]{P^{1/2} U} \big]^{2n} \big) &\leq 2^{2n-1} \big( \norm[0]{m}^{2n} + \expec\big[(U^\T P U)^n\big]\big) \\
&\leq 2^{2n-1} \big( \norm[0]{m}^{2n} + \norm[0]{P}^n \expec[\norm[0]{U}^{2n}] \big),
\end{split}
\end{equation*}
where $\expec(\norm[0]{U}^{2n})$ is the $n$th moment around zero of the chi-squared distribution with degrees of freedom $d$. That is,
\begin{equation*}
\expec(\norm[0]{U}^{2n}) = d \times \cdots \times (d+2(n-1)) \leq (d+2(n-1))^n \leq (d+2)^n n^n.
\end{equation*}
The inequality $(a+b)^{1/n} \leq a^{1/n} + b^{1/n}$ for any $a, b > 0$ yields the first claim. The second claim follows from $\expec(\norm[0]{X}^{2n}) = \expec([U^\T P U]^n)$ if $X$ is zero-mean.
\end{proof}

\section{Proof of \Cref{thm:ConInEq}}\label{appendix:ConcProofs}

This appendix contains the complete proof of \Cref{thm:ConInEq}.
We begin with a proposition providing bounds for functions satisfying certain differential inequalities. This proposition is a modification of Grönwall's inequalities~\eqref{eq:Gronwall1} and~\eqref{eq:Gronwall2} and appears also in~\cite[Appendix~A.2]{DelMoralKurtzmannTugaut2016}.

\begin{proposition}\label{prop:ODEProp} Let $\alpha \neq 0$ and $\beta \geq 0$ be constants and $n$ a positive integer. Suppose that a non-negative and differentiable function $x_t$ satisfies the differential inequality $\partial_t x_t \leq \alpha n x_t + \beta n^2 x_t^{1 - 1/n}$ for $t \geq t_0$. Then
\begin{equation*}
x_t^{1/n} \leq x_{t_0}^{1/n} \neper^{\alpha (t-t_0)} + \frac{\beta n}{\alpha} \big( \neper^{\alpha (t-t_0)} - 1 \big).
\end{equation*}
\end{proposition}

\begin{proof} For $t \geq t_0$, the function $z_t = \neper^{-\alpha n(t-t_0)} x_t \geq 0$ satisfies
\begin{align*}
\partial_t z_t = -\alpha n \neper^{-\alpha n(t-t_0)} x_t + \neper^{-\alpha n(t-t_0)} \partial_t x_t &\leq  n^2 \beta \neper^{-\alpha n(t-t_0)} x_t^{1 - 1/n} \\
&= \beta n^2 \neper^{-\alpha (t-t_0)} z_t^{1 - 1/n}.
\end{align*}
Consequently, for $t \geq t_0$, $\partial_t z_t^{1/n} = n^{-1} z_t^{1/n-1} \partial_t z_t \leq \beta n \neper^{-\alpha (t-t_0)}$ and direct integration yields
\begin{equation*}
z_t^{1/n} \leq z_{t_0}^{1/n} + \beta n \int_{t_0}^t \neper^{-\alpha (s-t_0)} \dif s = z_{t_0}^{1/n} + \frac{\beta n}{\alpha}\big( 1 - \neper^{-\alpha(t-t_0)} \big).
\end{equation*}
The claim is obtained by observing that $x^{1/n}_t = \neper^{\alpha  (t-t_0)} z^{1/n}_t$.
\end{proof}

\begin{proof}[Proof of \Cref{thm:ConInEq}] 

Application of Itô's lemma to~\eqref{eq:filtering-error-evolution} yields
\begin{equation*}
\begin{split}
\dif \norm[0]{E_t}^2 ={}& \dif M_t + \big[ \trace(Q) + \trace(S P_t^2) \big] \dif t \\
&+ 2 \big\langle f(X_t) - \Lin_{\Xkf_t,P_t}(f) - P_t S (X_t - \Xkf_t), X_t - \Xkf_t \big\rangle \dif t,
\end{split}
\end{equation*}
where 
\begin{equation*}
\dif M_t = 2 \inprod[\big]{Q^{1/2} \dif W_t - P_t H^\T R^{-1/2} \dif V_t}{X_t - \Xkf_t}
\end{equation*}
is a zero-mean (local) martingale. Keeping in mind that \sloppy{${\Lin_{\Xkf_t,P_t}(A) = A\Xkf_t}$} for any $A \in \R^{d_x \times d_x}$, we write
\begin{equation*}
f(X_t) - \Lin_{\Xkf_t,P_t}(f) - P_t S (X_t - \Xkf_t) = f(X_t) - P_t S X_t - \Lin_{\Xkf_t,P_t}\big(f - P_t S\big)
\end{equation*}
and apply \Cref{ass:Lt} to the function $g = f - P_t S$ with $x = X_t$ and $\tilde{x} = \Xkf_t$:
\begin{equation*}
\inprod[\big]{f(X_t) - \Lin_{\Xkf_t,P_t}(f) - P_t S (X_t - \Xkf_t)}{X_t - \Xkf_t} \leq -\lambda \norm[0]{X_t - \Xkf_t}^2 + C_\lambda \trace(P_t),
\end{equation*}
where $C_\lambda \geq 0$ is finite because $M(f-P_tS) \leq -\lambda$ and
\begin{equation*}
\begin{split}
N(f-P_tS) \geq N(f) + \nu(-P_tS) = N(f) - \mu(P_t S) &\geq N(f) - \norm[0]{P_t} \norm[0]{S} \\
&\geq N(f) - \trace(S) \lambda_P,
\end{split}
\end{equation*}
which is finite by \eqref{eq:JacobianBound} and the assumption $\sup_{t \geq 0} \trace(P_t) \leq \lambda_P$. For $t \geq T$, the assumption~\eqref{eq:contractivityCondition}, together with~\eqref{eq:lipBounds}, produces the almost sure bound
\begin{equation*}
\dif \norm[0]{E_t}^2 \leq - 2 \lambda \norm[0]{E_t}^2 \dif t + u \dif t + \dif M_t,
\end{equation*}
where $u = \trace(Q) + 2C_\lambda \lambda_P + \trace(S) \lambda_P^2$. Taking expectations and using Grönwall's inequality then yield the claimed mean square bound
\begin{equation}\label{eq:meanSquareBound}
\expec(\norm[0]{E_t}^2) \leq \expec(\norm[0]{E_T}^2) \neper^{-2\lambda(t-T)} + u/(2\lambda),
\end{equation}
with $\expec(\norm[0]{E_T}^2)$ being finite due to finiteness of $N(f)$ and $M(f)$. This concludes the proof of~\eqref{eq:MainMeanSquare}.

To prove the exponential concentration inequality~\eqref{eq:MainConcIneq} we compute upper bounds on $\expec(\norm[0]{E_t}^{2n})$ for every $n \geq 1$ in order to use \Cref{thm:bernstein}.
First, note that for $0 \leq t \leq T$ we have the inequality $\dif \norm[0]{E_t}^2 \leq 2\rho\norm[0]{E_t}^2 \dif t + u \dif t + \dif M_t$, where
\begin{equation*}
\rho = M(f) + \norm[0]{S} \trace(P_t) \geq M(f) + \mu(-P_t S) \geq M(f - P_t S).
\end{equation*}
In the following we can assume that $\rho > 0$, for if it were negative we could set $-\lambda = \rho$ and $T=0$. Let $\gamma$ stand for either $-\lambda$ or $\rho$. Observe that $\langle M \rangle_t$, the quadratic variation of $M_t$ (the increasing process such that $M_t^2 - \langle M \rangle_t$ is a martingale), satisfies
\begin{equation*}
\dif \, \langle M \rangle_t \leq 4 \norm[0]{E_t}^2 \big[ \trace(Q) + \trace( S P_t^2 ) \big] \dif t \leq 4 \norm[0]{E_t}^2 u \dif t.
\end{equation*}
For $n \geq 2$, the above inequality, the identity \sloppy{${\dif \, \langle \norm[0]{E}^2 \rangle_t = \dif \, \langle M \rangle_t}$}, and the general form of Itô's lemma then produce
\begin{equation*}
\begin{split}
\dif \norm[0]{E_t}^{2n} ={}& n \norm[0]{E_t}^{2(n-1)} \dif \norm[0]{E_t}^2 + \frac{n(n-1)}{2} \norm[0]{E_t}^{2(n-2)} \dif \, \langle \norm[0]{E}^2 \rangle_t \\
={}& n \norm[0]{E_t}^{2(n-1)} \dif \norm[0]{E_t}^2 + \frac{n(n-1)}{2} \norm[0]{E_t}^{2(n-2)} \dif \, \langle M \rangle_t \\
\leq{}& 2\gamma n \norm[0]{E_t}^{2n} \dif t + 2 u n^2 \norm[0]{E_t}^{2(n-1)} \dif t + n \norm[0]{E_t}^{2(n-1)} \dif M_t.
\end{split}
\end{equation*}
Induction establishes that $\expec(\norm[0]{E_t}^{2n})$ does not explode in finite time. Thus the term $\norm[0]{E_t}^{2(n-1)} \dif M_t$ vanishes when expectations are taken (see, e.g.,~\cite[Section 4.5]{KloedenPlaten1992} for similar arguments). Using Hölder's inequality with $p = n/(n-1)$, we get
\begin{equation*}
\begin{split}
\partial_t \expec(\norm[0]{E_t}^{2n}) &\leq 2\gamma n \, \expec(\norm[0]{E_t}^{2n}) + 2 u n^2 \, \expec(\norm[0]{E_t}^{2(n-1)}) \\
&\leq 2\gamma n \expec(\norm[0]{E_t}^{2n}) + 2 u n^2 \expec(\norm[0]{E_t}^{2n})^{1-1/n}.
\end{split}
\end{equation*}
We can now apply \Cref{prop:ODEProp} with $x_t = \norm[0]{E_t}^{2n}$ and $\beta = u$. Setting $\alpha = 2\rho$ and $t_0 = 0$ and considering $t \leq T$, we obtain
\begin{equation*}
\expec(\norm[0]{E_T}^{2n})^{1/n} \leq \bigg[ \expec(\norm[0]{E_0}^{2n})^{1/n} + \frac{un}{2\rho} \bigg] \neper^{2\rho T}.
\end{equation*}
Recall that $X_0 \sim \gauss(\mu_0, \Sigma_0)$ and $\Xkf_0 = \hat{x}_0$ is deterministic. Thus $E_0 \sim \gauss(\mu_0 - \hat{x}_0, \Sigma_0)$ so that \Cref{lemma:ChiSquareLemma} gives
\begin{equation*}
\begin{split}
\expec(\norm[0]{E_T}^{2n})^{1/n} &\leq \bigg[ 4\big( \norm[0]{\mu_0 - \hat{x}_0}^2 + \norm[0]{P_0}[d_x+2] n \big) + \frac{un}{2\rho} \bigg] \neper^{2\rho T} \\
&\leq \underbrace{ 4\bigg[ \norm[0]{\mu_0 - \hat{x}_0}^2 + \norm[0]{P_0}(d_x+2) + \frac{un}{8\rho} \bigg] \neper^{2\rho T} }_{\eqqcolon C_T} n.
\end{split}
\end{equation*}
This provides a bound on the initial value for the case $\alpha = -2\lambda$, $t_0 = T$, and $t \geq T$ in \Cref{prop:ODEProp}:
\begin{equation*}
\expec(\norm[0]{E_t}^{2n})^{1/n} \leq \expec(\norm[0]{E_T}^{2n})^{1/n} \neper^{-2\lambda(t-T)} + \frac{un}{2\lambda} \leq \bigg( C_T \neper^{-2\lambda(t-T)} + \frac{u}{2\lambda} \bigg) n.
\end{equation*}
That is,
\begin{equation*}
\expec(\norm[0]{E_t}^{2n}) \leq \Big( C_T \neper^{-2\lambda(t-T)} + u/(2\lambda) \Big)^n n^n.
\end{equation*}
The claimed exponential concentration inequality follows by applying Bernstein's inequality of \Cref{thm:bernstein} to $X = \norm[0]{E_t}^{2}$ with $\alpha = C_T \neper^{-2\lambda(t-T)} + u / (2\lambda)$.
\end{proof}

\section{Proof of \Cref{thm:ConInEqDisc}}\label{appendix:ConcProofDisc}

\begin{proof}[Proof of \Cref{thm:ConInEqDisc}]
Norm of the filtering error is
\begin{equation*}
\begin{split}
\norm[0]{E_k}^2 ={}& \big[ f(X_{k-1}) - \Lin_{\Xkf_{k-1},P_{k-1}}(f) \big]^\T ( \Id - K_k H )^\T \\
&\quad\quad\times( \Id - K_k H ) \big[ f(X_{k-1}) - \Lin_{\Xkf_{k-1},P_{k-1}}(f) \big] \\
&+ 2 \big[ f(X_{k-1}) - \Lin_{\Xkf_{k-1},P_{k-1}}(f) \big]^\T ( \Id - K_k H )^\T \\
&\quad\quad\times\big[ ( \Id - K_k H) Q^{1/2} W_k - K_k R^{1/2} V_k \big] \\
&+ \big[ ( \Id - K_k H) Q^{1/2} W_k - K_k R^{1/2} V_k \big]^\T \big[ ( \Id - K_k H) Q^{1/2} W_k - K_k R^{1/2} V_k \big].
\end{split}
\end{equation*}
We immediately obtain
\begin{equation*}
\kappa = \sup_{k \geq 1} \norm[0]{K_k} = \norm[1]{P_{k \mid k-1} H^\T (H P_{k \mid k-1} H^\T + R)^{-1}} \leq \lambda_P^p \norm[0]{H} \norm[0]{R^{-1}}.
\end{equation*}
Using \Cref{ass:LtDisc} and \eqref{eq:contractivityConditionDisc}, we get the recursive filtering error bound
\begin{equation} \label{eq:recursive-error-disc}
\norm[0]{E_k}^2 \leq \lambda_{df}^2 \norm[0]{E_{k-1}}^2 + C_f \lambda_d^2 \lambda_P^u + \norm[0]{U_k}^2 + 2 M_k,
\end{equation}
where $C_f$ is the constant of \Cref{ass:LtDisc} for the function $f$ and the random variables
\begin{equation*}
  U_k = ( \Id - K_k H) Q^{1/2} W_k - K_k R^{1/2} V_k
\end{equation*}
and
\begin{equation*}
M_k = \big[ f(X_{k-1}) - \Lin_{\Xkf_{k-1},P_{k-1}}(f) \big]^\T ( \Id - K_k H )^\T U_k
\end{equation*}
are zero-mean because $W_k$ and $V_k$ are independent of $X_{k-1}$, $\Xkf_{k-1}$, $P_{k-1}$, and $K_k$. 
Because
\begin{equation} \label{eq:UkVarBound}
\expec(\norm[0]{U_k}^2) \leq u_d \coloneqq \lambda_d^2 \trace(Q) + \kappa^2 \trace(R),
\end{equation}
we have 
\begin{equation*}
\expec( \norm[0]{E_k}^2 ) \leq \lambda_{df}^2 \expec( \norm[0]{E_{k-1}}^2 ) + u_d + C_f \lambda_d^2 \lambda_P^u.
\end{equation*}
The discrete Grönwall's inequality~\eqref{eq:Gronwall-discrete} then produces
\begin{equation*}
\expec( \norm[0]{E_k}^2 ) \leq \lambda_{df}^{2k} \expec( \norm[0]{E_0}^2 ) + \frac{u_d + C_f \lambda_d^2 \lambda_P^u}{1-\lambda_{df}^2} = \lambda_{df}^{2k} \big( \norm[0]{\mu_0 - \hat{x}_0}^2 + \trace(\Sigma_0) \big) + \frac{u_d + C_f \lambda_d^2 \lambda_P^u}{1-\lambda_{df}^2},
\end{equation*}
which is the mean square bound~\eqref{eq:MainMeanSquareDisc}.
Here we have used $\expec( \norm[0]{X}^2 ) = \trace(P) + \norm[0]{m}^2$, which holds for any Gaussian random vector $X \sim \gauss(m, P)$.

To obtain the concentration inequality~\eqref{eq:MainConcIneqDisc}, we bound $\expec( \norm[0]{E_k}^{2n} )$ for every $n \geq 1$ and use Bernstein's inequality.
The random variable $M_k$ admits the bound
\begin{equation*}
M_k \leq \big( \lambda_{df} \norm[0]{E_{k-1}} + \eta \big) \norm[0]{U_k} \quad \text{ with } \quad \eta = \lambda_d(C_f \lambda_P^u)^{1/2}.
\end{equation*}
Inequality~\eqref{eq:recursive-error-disc} gives
\begin{equation*}
\expec( \norm[0]{E_k}^{2n} ) \leq \expec\big[ \big( \lambda_{df}^2 \norm[0]{E_{k-1}}^2 + \norm[0]{U_k}^2 + 2 M_k + \eta^2 \big)^n \big]
\end{equation*}
and consequently Minkowski's inequality yields
\begin{equation}\label{eq:DiscEK2n1n}
\begin{split}
\expec( \norm[0]{E_k}^{2n} )^{1/n} &\leq \expec\big[ \big( \lambda_{df}^2 \norm[0]{E_{k-1}}^2 + \norm[0]{U_k}^2 + 2 M_k + \eta^2 \big)^n \big]^{1/n} \\
&\leq \lambda_{df}^2 \expec( \norm[0]{E_{k-1}}^{2n} )^{1/n} + \expec( \norm[0]{U_k}^{2n} )^{1/n} + 2 \expec( M_k^n)^{1/n} + \eta^2.
\end{split}
\end{equation}
By \Cref{lemma:ChiSquareLemma} and \eqref{eq:UkVarBound}, $\expec( \norm[0]{U_k}^{2n} )^{1/n} \leq (d_x + 2) n u_d$, and by Minkowski's and Hölder's inequalities,
\begin{equation*}
\begin{split}
\expec( M_k^n)^{1/n} &\leq \expec( \norm[0]{U_k}^{n} )^{1/n} \expec \big( [\lambda_{df} \norm[0]{E_{k-1}} + \eta]^n \big)^{1/n} \\
&\leq \expec( \norm[0]{U_k}^{n} )^{1/n} \big[ \lambda_{df} \expec(\norm[0]{E_{k-1}}^n)^{1/n} + \eta \big] \\
&\leq \expec( \norm[0]{U_k}^{2n} )^{1/(2n)} \big[ \lambda_{df} \expec(\norm[0]{E_{k-1}}^{2n})^{1/(2n)} + \eta \big] \\
&\leq \sqrt{(d_x + 2) n u_d} \big[ \lambda_{df} \expec(\norm[0]{E_{k-1}}^{2n})^{1/(2n)} + \eta \big].
\end{split}
\end{equation*}
Inserting these bounds into \eqref{eq:DiscEK2n1n} and recognizing that the result can be bounded by a sum of two quadratic terms produces
\begin{equation*}
\expec( \norm[0]{E_k}^{2n} )^{1/(2n)} \leq \lambda_{df} \expec(\norm[0]{E_{k-1}}^{2n})^{1/(2n)} + 2\sqrt{(d_x + 2) n u_d} + \eta.
\end{equation*}
Then the discrete Grönwall's inequality and \Cref{lemma:ChiSquareLemma} yield
\begin{equation*}
\begin{split}
\expec( \norm[0]{E_k}^{2n} )^{1/(2n)} &\leq \lambda_{df}^k \expec( \norm[0]{E_0}^{2n} )^{1/(2n)} + \frac{2\sqrt{(d_x + 2) n u_d} + \eta}{1 - \lambda_{df}} \\
&\leq \bigg( 2 \lambda_{df}^k \Big[\norm[0]{\mu_0 - \hat{x}_0} + \norm[0]{\Sigma_0}^{1/2} \sqrt{(d_x+2) n} \, \Big] + \frac{2\sqrt{(d_x + 2) n u_d} + \eta}{1 - \lambda_{df}} \bigg) \\
&\leq 2 \bigg( \lambda_{df}^k \Big[\norm[0]{\mu_0 - \hat{x}_0} + \norm[0]{\Sigma_0}^{1/2} \Big] + \frac{\sqrt{u_d} + \eta}{1 - \lambda_{df}} \bigg) \sqrt{(d_x+2) n}.
\end{split}
\end{equation*}
This proves that
\begin{equation*}
\begin{split}
\expec( \norm[0]{E_k}^{2n} ) &\leq 2^{2n} \bigg( \lambda_{df}^k \Big[\norm[0]{\mu_0 - \hat{x}_0} + \norm[0]{\Sigma_0}^{1/2} \Big] + \frac{\sqrt{u_d} + \eta}{1 - \lambda_{df}} \bigg)^{2n} (d_x + 2)^n n^n
\end{split}
\end{equation*}
The claim follows from Bernstein's inequality with $X = \norm[0]{E_k}^2$ and
\begin{equation*}
\alpha = 4 \bigg( \lambda_{df}^k \Big[\norm[0]{\mu_0 - \hat{x}_0} + \norm[0]{\Sigma_0}^{1/2} \Big] + \frac{\sqrt{u_d} + \eta}{1 - \lambda_{df}} \bigg)^{2}.
\end{equation*}

\end{proof}

\newpage

\bibliographystyle{siamplain}

\end{document}